\documentclass[12pt]{amsart}

\usepackage{amssymb}
\usepackage{amsthm}
\usepackage{amsmath}
\usepackage{amsxtra}
\usepackage{mathrsfs}
\usepackage{accents}
\usepackage{ stmaryrd }
\usepackage{graphicx}
\graphicspath{ {images/} }
\usepackage[lmargin=1in,rmargin=2in,marginpar=2in,vmargin=1.5in]{geometry}
\usepackage{enumerate}
\usepackage{verbatim}
\usepackage[font=small,labelfont=bf]{caption}

\newcounter{count}[section]

\newtheorem{theorem}[count]{Theorem}
\newtheorem{lemma}[count]{Lemma}
\newtheorem{proposition}[count]{Proposition}

\theoremstyle{definition}
\newtheorem{definition}[count]{Definition}
\newcommand{\graphheight}{120}

\usepackage[linesnumbered,vlined,ruled]{algorithm2e}

\usepackage{fullpage}

\begin{document}
\title{Two Algorithms for the Package-Exchange Robot-Routing Problem}
\date{written June 26, 2015}
\author{James Drain}
\address{Department of Mathematics 
\\Dartmouth College}
\email{james.r.drain.17@dartmouth.edu}

\maketitle

\begin{abstract} We present and analyze two new algorithms for the package-exchange robot-routing problem (PERR): restriction to inidividual paths (RIP) and bubbletree. RIP provably produces a makespan that is $O(\text{SIC}+k^2)$, where SIC is the sum of the lengths of the individual paths and $k$ is the number of robots. Bubbletree produces a makespan that is $O(n)$, where $n$ is the number of nodes. With optimizations bubbletree can also achieve a makespan of $O((k+l)\text{log}k)$, where $l$ is the longest path from start to goal in the bubbletree subgraph.
\end{abstract}
\section{Statement of the Package-Exchange Robot-Routing Problem}
Let the undirected, finite graph $G$ have sets of start and goal nodes $\{s_1,\dots,s_k\}$ and $\{g_1,\dots,g_k\}$. Initially there is one phantasmagorical ghost $r_i$ at each $s_i$. In each timestep, the ghosts are allowed to either stay at their current node or move to a neighboring node. No two ghosts can occupy the same node at the same timestep, but they are allowed to pass through each other and swap nodes. The problem is to find a sequence of moves that makes it so that all the ghosts are at their goals at as early a timestep as possible. For a given sequence of moves, that first time when the ghosts are all at their goals is known as the makespan.
\section{Applications}
There are a few, less ethereal applications. We can replace the $k$ ghosts with $k$ transferable payloads carried by as many robots: packages are allowed to be exchanged between neighboring robots, but no robot can hold two packages simultaneously. Indeed, we can then move the packages along exactly as we would the ghosts, with the understanding that the packages are being carried. We can imagine applications to ride-sharing (or taxis) with passenger transfers (Coltin and Veloso 2014) or package delivery with robots in offices
(Veloso et al. 2015), although it is unlikely for the conditions of the problem to be exactly met. %taken from Hang's paper. Also, the taxi problem is totally different...
%Throughout the paper I will drop the ridiculous ghost-talk and refer to the agents as robots.
\par PERR can also be seen as a relaxed form of the multi-agent path-finding problem, MAPF. MAPF is identical to PERR, except that swaps are disallowed. Not all instances of MAPF are solvable; for example MAPF is unsolvable on a graph with only two nodes and two agents that need to change places. In contrast, PERR is always solvable, and both of this paper's algorithms solve all instances of PERR: they are complete. MAPF is a well-studied problem. It was first stated as the pebble-passing problem (Kornhauser 1984), whereas PERR was first posed in 2016 (Hang). The current state of the art optimal solver for MAPF is CBS (Sharon et al. 2013), which can easily be modified to apply to PERR. %maybe also cite the flow methods from Hang's paper
\par Even with these assumptions and relaxations, PERR (and MAPF) are NP-hard to approximate to within a factor of $4/3-\epsilon$ (Hang 2016) More specifically, we can encode variants of $3-\text{SAT}$ as instances of PERR (and MAPF) with makespans $3$ or $4$ depending on whether the corresponding boolean expression is satisfiable. %ugh, I want to say that the best I hope for is big-O optimality, but, like, the proof for the $4/3$ approximation thing was only for makespan = 3...
Thus we turn to polynomial algorithms that are asymptotically optimal to within a cruder factor. 
\par The restriction to individual paths (RIP) algorithm is an intuitive algorithm that is empirically near-optimal on graphs with relatively few robots. %It performs slightly worse than parallel bubblesort on random linear test arrays, and it performs significantly worse than shearsort,  cousin of bubblesort, on random square arrays.
We also present the bubbletree algorithm, which takes its inspiration from the proof of parallel bubblesort. The basic bubbletree produces provably bounded-suboptimal results for low degree, robot-dense tree graphs with many robots in which the longest individual path length is proportional to the number of robots. More precisely, it gives a makespan that is $O(dn)$, where $n$ is the number of nodes and $d$ is the maximum degree. An optimization produces an $O(n)$ and $O(\text{diam}*\text{log}k+k)$ makespan for all graphs, where $k$ is the number of robots and $\text{diam}$ is the diameter of the graph.%$O(l+k)$ makespan on trees, which is best possible. 
There are however families of graphs that cannot be restricted to a tree while maintaining near-optimal makespans. Incidentally, RIP performs perfectly on the counterexamples we present.

\section{Restriction to Individual Paths (RIP) Algorithm}
The idea of RIP is to form optimal paths $P_{0,i}$ for each robot at timestep $0$ and naively advance them along their paths if possible. Let $P_{t,i}$ denote the shortest individual path for $r_i$ to follow to reach its goal node starting at time $t$. This can easily be accomplished using $A^\ast$, avoiding the combinatorial explosion of planning for many robots. We can also index into paths. Thus $P_{t,i,0}$ denotes $r_i$'s location at time $t$, and $P_{t,i,1}$ denotes the node $r_i$ wants to go to at timestep $t$. When unambiguous, we abuse notation and conflate robots with their indices.
\par \begin{definition}[advance] We say that robot $r_a$ \textit{advances} at time $t$ if $r_a$ moves towards its goal along its shortest path at time $t$ i.e. $P_{t+1,a}\subsetneq P_{t,a}$.
\end{definition} If we advance robot $r_i$ at time $t$, then we delete $P_{t,i,0}$ from $P_{t,i}$ to form $P_{t+1,i}$. If there is a cycle of robots $r_1,\dots,r_m$ such that $r_i$ wants to move to the node occupied by $r_{i+1\pmod{m}}$, then we advance those robots simultaneously in a cycle. If $m=2$, then this is a happy swap. 
\begin{definition}[happy swap] Say that $a$ and $b$ \textit{happy swap} if they both decrease the lengths of their shortest paths when they swap with each other. %in notation, $a\smil b$
\end{definition}
\par Happy swaps contrast with bully swaps.
\par \begin{definition}[bully swap] Say robot $a$ \textit{bully-swaps} $b$ if, when $a$ swaps with $b$, then $a$ shortens its own shortest path, but lengthens $b$'s shortest path. %in notation, $a \frown b$. %but lengthens $b$'s shortest path.  -- yeah, I can use either lengthens or doesn't shortest. They're interchangeable for my paper
% \par When I say $a$ swaps past $b$, or $b$ is swapped by $a$, I mean that $a$ bully-swaps $b$. I denote a generic swap by saying $a$ and $b$ swap.
%define swaps (in particular swapping vs. being swapped (this is when we have P_x\prec P_y) -- emphasize that we can have paths $(1,2,3,4) \subset (4,3,2,1)$, without $(1,2,3,4)\prec (4,3,2,1), say
\end{definition}
%In RIP, we perform all swaps $r_a$ past $r_b$ in which $P_{b,t}$ is a substring of $P_{a,t}$.
%\par \begin{definition}[$P_b \prec P_a$] Say $P_b\prec P_a$ if $P_b$ is a substring of $P_a$ i.e. if all the nodes in $P_b$ appear consecutively in $P_a$ in the same order as in $P_b$.
%\end{definition}
%\par We refer to the swaps $x\leftrightarrow y$ for which $P_{b,t}\prec P_{a,t}$ as substring swaps. 
In RIP, we bully swap $r_a$ past $r_b$ only if $P_{t,b}$ is a subset of $P_{t,a}$.
These are known as subset swaps.
 \par \begin{definition}[subset-swap] Say that $a$ \textit{subset-swaps} $b$ if $P_{b,t}\subsetneq P_{a,t}$ and $a$ swaps with $b$. \end{definition}%bully-swaps $b$. %these are prototypically bully swaps This is a bully-swap unless $a$ immediately pulls $b$ to $g_b$, as I will show after defining the notion of parallel paths.  %, as I will show in lemma 2 of the makespan proof for RIP.

% \par We refer to the swaps $x\leftrightarrow y$ for which $P_{y,t}\subset P_{x,t}$ as subset-swaps.
%\par \begin{lemma} Let $P_{t,b}\prec P_{t,a}$ be shortest paths. If $a$ substring-swaps $b$, then this is a bully swap.% More strongly, this increases the length of $b$'s shortest path to $g_b$.
%\end{lemma}
%\par \begin{proof}
%Observe that every substring of a shortest path is also a shortest path. After substring-swapping $a$ past $b$, $b$'s new shortest path $P_{t+1,b}$ is then the prefix of $P_{t,a}$ terminating at $g_b$, %which extends $P_{t,b}$.
%\end{proof}

%\par What does it mean to swap-need?
Further define directed swap.
\begin{definition}[(directed) swap need $(a,b)$] Say that $a$ swap-needs $b$ if, after removing all other robots, it is impossible to advance $a$ and $b$ to their goals without $a$ bully-swapping $b$.
\end{definition}
The advantage of subset swaps $a\leftrightarrow b$ is that they are guaranteed to resolve swap needs $(a,b)$ when $P_{t,b}$ is a parallel subset of $P_{t,a}$ i.e. the nodes in $P_{t,b}$ appear in the same order in $P_{t,a}$. In this case we write $P_{t,b}\prec P_{t,a}$.
% If robot $j$ is bully-swapped at time $t$ by robot $i$ (by one of these substring-swaps),
If robot $r_j$ is subset-swapped at time $t$ by robot $r_i$, then we prepend $P_{t,i,0}$ to $P_{t,j}$ to form the shortest path $P_{t+1,j}$. %I will show using the notion of parallel paths that each $P_{i,t}$ is indeed a shortest path.
\par In the pseudocode for RIP, the inner while loop advances robots that can move without cycling or swapping. We then perform subset swaps. Finally we resolve any cycles of unmoved robots. (RIP is complete for any ordering of these three operations.)
%Thus we initially mark all robots as unmoved at time $t$, and mark them as and update their paths when we commit to moving them. 
\par In general, robots move simultaneously in a given timestep, but we have to plan serially which robots to move. For example, it is possible that $r_x$ occupies $P_{y,t,1}$ at time $t$, but we determine while looping through the unmoved robots that $r_x$ will move now, therein opening a spot for $r_y$ to move to. Thus, in a given timestep, we continuously loop through the unmoved robots until going a loop without moving any new robots. $Q_{t,i}$ denotes the location of $r_i$ at time $t$ in the global product of paths that we find with our algorithm.%For ease of implementation, we move robots serially within a timestep. But of course we would move all robots simultaneously during that timestep if executing the solution in real life. %if $P_{x,t} = pos(x,t)P_{y,t}Z$, for some 
\begin{algorithm}[!htbp]
\small
\caption{RIP}
\label{Detail} %%?
	$t \gets 0$\;
    \For{i=1,\dots,k}
    {
    	$P_{0,i} = shortestPath(s_i,g_i)$\;
    }
    \For{i=1,\dots,k}
    {
        $Q_{t,i}\gets P_{t,i,0}$\;
    }
    \While{some agent is not at its goal}
    {
		mark all robots as unmoved\;
        movedAnyone = true\;
		\While{movedAnyone}
        {
        	movedAnyone $\gets$ false\;
        	\For{$i=1,\dots,k$}
            {
            	\If{$P_{t,i,1}$ is unoccupied}
                {
                    mark $r_i$ as moved\;
                    movedAnyone $\gets$ true\;
%                     $Q_{i,t} \gets v$\;
                    delete $P_{t,i,0}$\;% to form $P_{t+1,i}$\;
                }
            }
        }
 	\For{$i=1,\dots,k$}
  	{
    	\If{$r_i$ is unmoved}
        {
    	\If{there is an unmoved robot $r_j$ occupying $P_{t,i,1}$}
        {
          \If{$P_{t,j}\prec P_{t,i}$}
          {
          %swap $r_o$ and $r_a$
          mark $r_i$ and $r_j$ as moved\;
%           $Q_{i,t} \gets v$\;
%           $Q_{i,t} \gets P_{i,t,0}$\;
          prepend $P_{t,i,0}$ to $P_{t,j}$ to form $P_{t+1,j}$\;
          delete the zeroth entry of $P_{t,i}$ to form $P_{t+1,i}$\;
          }
    	}
       }
   	}
        resolve any cycles of unmoved robots by moving all involved robots in a circle, updating the $P$'s appropriately and marking those robots as moved\;
%         \For {each unmoved robot $r_i$}
        \For {$i=1,\dots,k$}
        {
            $P_{t+1,i} \gets P_{t,i}$\;
        	$Q_{t+1,i}\gets P_{t+1,i,0}$\;
        }
        $t \gets t+1$\;
}
Return $Q$\;

\end{algorithm}

% \subsection{Details of RIP Algorithm}

% \par Finding the shortest path for each agent
% %     We want to make it so that each pair of robots will ever contest at most one bottleneck node. That is, we want to make it so that there's no pair of paths $P_1\supset a,x_1,\dots,x_n,b$ and $P_2\supset a,y_1,\dots,y_n,b$ with some $y_i\neq x_i$. This can be done by first finding shortest paths between all pairs of points in such a way that longer paths extend canonical shortest paths... I think; this would also make it easy to look up the next move of each agent -- all we need is its location and goal. Alternatively, we can first find a shortest path for each agent and then loop over each pair and eliminate bad substrings, again taking $O(n^3)$ steps.
% \par Finding the bottleneck sets: This is easy.
% \par Finding cycles: this is easy. Start with the first active robot $r$ and iteratively follow $succ(r)$ until looping or hitting an inactive robot.

\section{Proofs of Completeness and $O(k^2+\text{SIC})$ Makespan for RIP}
% Define robot $i$'s swap set at time $t$ to be $S_{i,t}:= \{j\in [k]:P_{j,t}\subsetneq P_{i,t}\}$ i.e. the indices of the robots that $r_i$ would have to swap past if no other swaps occur. We say that $r_i$ contains $r_j$ (at time $t$) if $j\in S_{i,t}$. %Similarly, define the swapped-past set $S^\supset_i:= \{j\in [k]:P_j\supset P_i\}$. %Define $S = \sqcup_{i=1}^kS_i$ (thus $|S| = \sum_{i=1}^k|S_i|$.) 
% Define $S_t = \sum_{i=1}^k|S_{i,t}|$.
% Define the potential $(\sum_{i=1}^k d(r_i,g_i),S_t)$ 
% and order it lexicographically i.e. $(a,b)>(c,d)$ if and 
% only if $a>c$ or $a=c$ and $b>d$. Observe that $S\leq 2k^2$ and $\sum_{i=1}^k d(r_i,g_i)\leq kn$, where 
% $n$ is the number of vertices. I claim that the potential 
% strictly decreases with every timestep (thus the algorithm 
% terminates in at most $k^3n$ timesteps... which is pretty bad, but finite.)
% Here $l$ denotes the length of the longest individual shortest path to goal i.e. $l = \text{max}_{i=1,\dots,k}d(s_i,g_i)$.
Here $\text{SIC} = \sum_{i=1}^kd(s_i,g_i)$ is the sum of individual costs heuristic. Since PERR is discretized, the distance $d(s_i,g_i)$ is simply the least number of edges on a path from $s_i$ to $g_i$.
\par We first need to verify that, if not all robots are at their goal, then at least one robot moves. Define a robot $r_i$'s \textit{bottleneck set} $B_i$ at time $t$ to be the set of robots that want to advance to the same node as $r_i$. Define $f:[k]\rightarrow [k]\cup V$ by letting $f_t(i)$ denote the (index of) the robot occupying $P_{i,t,1}$ if there is one, and otherwise letting it denote $P_{i,t,1}$ itself. At the beginning of timestep $t$, start at any robot $r_a$ that is not at its goal, and follow the chain $r_a,f_t(a),f_t^{(2)}(a),\dots, f^{(m)}_t(a)$ until first reaching an unoccupied node in the graph or else repeating a robot. If this chain terminates with $f^{(m)}_t(a)\in V$, then either the robot $f^{(m-1)}_t(a)$ will claim that vertex, or else some other robot in its bottleneck set will claim that vertex. If the chain terminates with $f^{(m)}_t(a)$ equaling an earlier $f^{(m-i)}_t(a)$ with $i>1$, then either some subset of $f^{(m-i)}_t(a),\dots,f^{(m)}_t(a)$ will move before detecting cycles, or else they will all move in a cycle. Finally, if $f^{(m)}_t(a) = f^{(m-1)}_t(a)$, then some robot in $f^{(m-2)}_t(a)$'s bottleneck set $B$ will move: some robot in $B$ will swap past $f^{(m)}_t(a)$, or else every robot in $B$ will be swapped.
\par Now that we know some robot moves every timestep, let us define the notion of swap risk.
\par\begin{definition}[(directed) swap risk] Say there is a \textit{(directed) swap risk} $(a,b)$ if there is already a swap need $(a,b)$ or if, after removing all other robots, it is possible to introduce a swap need $(a,b)$ by advancing $a$ and $b$ along their respective shortest paths.
\end{definition}
\par The proof of completeness works by defining the potential function $\Phi = \sum_{i=1}^k d(r_i,g_i) + \sum_{i=1}^k\sum_{j=1}^k s(i,j)$, where $s(i,j) = 1$ if $i$ swap-risks $j$, and otherwise $s(i,j) = 0$ (we always have $s(i,i) = 0$.)  Observe that $\Phi \leq k^2+\text{SIC}$. It is sufficient to show that a single action decreases $\Phi$ by at least $1$; performing $m$ actions simultaneously will thus decrease the potential by at least $m$.
% \par What is a swap risk?

%\par Swap needs are those swap risks for which the pair of robots can incur a swap need by advancing zero steps.
\par By the definition of swap risk, neither advancing a robot along its unimpeded path nor resolving a cycle of robots introduces a swap risk. Advancing an unimpeded robot and resolving a cycle decrease the sum of distances. Thus both actions decrease the potential and we only have to consider bully swaps, which we know do not increase the sum of distances.
\par Before continuing, we need to introduce the notion of parallel paths.
\par \begin{definition}[parallel paths] Let $P_{t,c} \cap P_{t,d} = \{a_0,\dots,a_n\}$ be nonempty and let those nodes appear in the same order $a_0,\dots,a_n$ in both $P_{t,c}$ and $P_{t,d}$. Then $P_{t,c}$ and $P_{t,d}$ are said to be parallel paths and we write $P_{t,c} \rightrightarrows P_{t,d}$ (or, equivalently, $P_{t,d} \rightrightarrows P_{t,c}$). Conversely, if $n\geq 1$ and the nodes appear in that order in $P_{t,c}$ but in opposite order in $P_{t,d}$, then we say the two paths are antiparallel and write $P_{t,c} \rightleftarrows P_{t,d}$ (or, equivalently, $P_{t,d} \rightleftarrows P_{t,c}$).  If $n=1$, the paths are more specifically trivially parallel.
\end{definition}
\par \begin{lemma} Shortest paths are either parallel or antiparallel (or nonintersecting).
\end{lemma}%comment that we get transitivity of antiparallelism by inclusion, but, in general, neither parallelism nor antiparallelism is transitive
\par \begin{proof} Let $P_{t,c}$ and $P_{t,d}$ be two shortest paths and let $\{x,y,z\}$ appear in both paths, and let $x,y,z$ appear in that order in $P_{t,c}$. It is sufficient to show that $\{x,y,z\}$ must appear in the same or opposite order in $P_{t,d}$; then we could inductively show that the whole paths are (anti-)parallel by showing that the first three common nodes of $P_{t,c}$ have (anti-)parallel counterparts, then common nodes two through four must also have (anti-)parallel counterparts, then common nodes three through five etc. %$P_{t,c} \cap P_{t,d} = \{a_0,\dots,a_n\}$, where  those nodes appear in the order $a_0,\dots,a_n$ in $P_{t,c}$, but in neither that order nor the reverse order in $P_{t,d}$. Then we can find some triple of nodes, wlog $a_0,a_1, \text{ and } a_2$ that appear
Thus, for the sake of contradiction and wlog, assume that $x$ comes between $y$ and $z$ in $P_{t,d}$. By the shortest paths hypothesis applied to $P_{t,d}$, we have $d(x,z) < d(y,z)$. But applying the shortest paths hypothesis to $P_{t,c}$ gives $d(y,z)<d(x,z)$, contradiction. 
\end{proof}
\par Note that if two shortest paths are antiparallel and we extend one of them to another shortest path, then they will still be antiparallel. Conversely, if two paths are parallel, then we can remove extremal nodes while maintaining parallelism. However, neither parallelism nor antiparallelism is transitive in general. For example, let $P_1 = a,b,c$, $P_2 = a,b,d$, and $P_3 = d,b,c,e$. Then $P_1\rightrightarrows P_2$ and $P_1\rightrightarrows P_3$, but $P_2 \leftrightarrows P_3$.
\par Parallel paths give a better sense of what subset swapping does in RIP. If $P_a \leftrightarrows P_b$, then the only way $a$ can subset-swap $b$ is if $P_b$ has only two nodes: the node $v_b$ occupied by $b$, and the node $g_b$ occupied by $a$. In this case, $a$ immediately brings $b$ to its goal. If $P_b\rightrightarrows P_b$ and $a$ subset-swaps $b$, then this is a bully swap as $P_{t+1,b}\subset P_{t,a}$ is a shortest path, and longer than $P_{t,b}$.
\par Let $a$ bully swap $b$ and let $x\neq a,b$. When $a$ subset-swaps $b$, they resolve the swap risk $(a,b)$. This is proved after formally characterizing swap risks in proposition \ref{Swap Risk Characterization}. %need parallel paths to handle $(b,a)$

% It is thus sufficient to show that this swap doesn't increase $\sum_{i=1}^ks(x,i)+s(i,x)$, as $\sum_{i=1}^k\sum_{j=1}^k s(i,j)= \sum_{x\neq a,b}\sum_{i=1}^ks(x,i)+s(i,x)-s(a,b)-s(b,a)$. %that doesn't seem right...
Clearly no swap risk $(x,y)$ is introduced for $y\neq a,b$. Since $a$ advances, there is no new swap risk between $x$ and $a$, either. Since $g_a$ is not on $P_{t+1,b}$, there is also no new swap risk $(b,a)$ (We will prove this after characterizing swap risks.)
%Finally, assume that there is a new swap risk between $x$ and $b$. 
It is thus sufficient to prove the following lemma:
\par \begin{lemma}[Main Monovariant Lemma] The only way bully subset-swapping $a$ past $b$ can incur a new swap risk between $x$ and $b$ is if it is a directed swap risk $(b,x)$, and if swapping this way simultaneously resolves a swap risk $(a,x)$.
\end{lemma}
\par The proof requires a couple lemmas and a definition.
% \par Lemma 1. Let $P_{t,c}$ and $P_{t,d}$ be shortest paths and let $c$ subset-swap $d$ (recall that happy swaps, for example if $P_{t,c} = (v,w)$ and $P_{t,d} = (w,v)$, do not count as subset-swaps.) Then $P_{t,c}$ and $P_{t,d}$ are parallel.
% \par Proof. 
\par \begin{lemma} RIP maintains the shortest paths invariant. That is, every $P_{t,a}$ is at all times a shortest path from $P_{t,a,0}$ to $g_a$.
\end{lemma}
\par \begin{proof} Assume that all paths are shortest paths at time $t$. Clearly advancing a robot along its path keeps it on a shortest path. Let $a$ subset-swap $b$. Let $P_{t,b} = v_0,v_1,\dots,v_n$, (with those vertices appearing in that order). If $n=1$, then it is possible that the two paths are antiparallel, in which case the swap is a happy swap and the robots advance, maintaining the shortest paths invariant. This is the only case in which the paths are antiparallel: if they are antiparallel, then we need $v_n$ to appear in $P_{t,a}$ before $v_0,v_1,\dots, \text{ and } v_{n-1}$. But we know that $v_0 = P_{t,b,0} = P_{t,a,1}$. Thus $P_{t,a}$ begins $v_n,v_0$, hence $n=1$. Now consider when the two paths are parallel. If $b$ is at its goal before being swapped, then its path after being swapped is length one, which is best possible. If $b$ is not at its goal, then $P_{t+1,b} = P_{t,a,0},v_0,v_1,\dots,v_n$ with $n>1$. Since $P_{t,a,0},v_0,v_1,\dots,v_n$ is a shortest (sub)path from $P_{t,a,0}$ to $v_n$, it follows that $P_{t+1,b}$ is indeed a shortest path.
\end{proof}% is a subset of the shortest path $P_{t,a}$, $P_{t+1,b}$ is indeed a shortest path.%ugh, I should call this subseq-swap instead of subset-swap
% \par Lemma 3. In RIP, in the case of a swap risk between robots $c$ and $d$, there is exactly one robot that may need to swap past the other.%say that their paths are shortest?
% \par Proof. There cannot be a swap risk between antiparallel paths $P_{t,c}$ and $P_{t,d}$, because we could always alternate advancing $c$ and $d$ as far as possible until one of them is blocked, say $d$ blocks $c$. Let $v\neq g_d$ be common to both paths. It is impossible for $d$ to be at its goal, because  %the two robots until immediately before a common node, and then happy swap them. UGH, but how do I know this is a happy swap? 

%what's the point of this lemma? To show that directed swap risks are a thing?
\par In the following lemmas, all paths are shortest paths, as appear in RIP.
\par \begin{lemma} If $P_c \rightrightarrows P_d$ are parallel shortest paths, then there is at most one of the swap risks $(c,d)$ and $(d,c)$, i.e. parallel swap risks are inherently directed. If $P_c \leftrightarrows P_d$ are antiparallel shortest paths, then we can have any combination of swap risks $(c,d)$ and $(d,c)$.
\end{lemma}
\par \begin{proof} Let $P_c \rightrightarrows P_d$. %Let the common nodes be $v_0,\dots,v_n$, in that order.
The proof uses the fact that we cannot simultaneously have $g_d\in P_c$ and $g_c \in P_d$; otherwise we would have $g_c,g_d$ appear in that order in $P_d$, and in reverse order in $P_c$. For two robots on parallel paths to incur a swap need, it must inevitably get to the point where one of the robots, say $c$, is not yet at its goal, but its next node is occupied by $d$, who is at its goal. By the fact we just mentioned, that means we have at most one of $s(c,d)$ and $s(d,c)$. Now let $P_c \leftrightarrows P_d$. If $P_c = v_0,v_1$ and $P_d = v_1,v_0$, then we have neither $s(c,d)$ nor $s(d,c)$. If $P_c = v_0,v_1,v_2$ and $P_d = v_3,v_2$, then we can incur a swap need by advancing $P_d$ one step to $v_2$. If $P_c = v_0,v_1,v_2,v_3$ and $P_d = v_4,v_3,v_2,v_1$, then we can incur a swap need $(c,d)$ by advancing $d$ to its goal, or a swap need $(d,c)$ by advancing $c$ to its goal.
\end{proof}
%  \par \begin{lemma} %If $g_d \notin P_c\setminus P_{c,0}$, then $c$ does not swap need $d$.
%  If $g_d$ is not on $P_c$, then $c$ does not swap-need $d$.
% \end{lemma}
% \par \begin{proof} The only way $c$ and $d$ can reach a standstill that requires $c$ to swap $d$ is if $d$ has reached its goal, which $c$ currently wants to move to. This is impossible if $g_d$ is not on $P_c$.%$g_d \notin P_c\setminus P_{c,0}$.
% \end{proof}
\par We have the following characterization of swap-needs:
%%%%This seems like an important characterization, but I don't use it...
\par \begin{proposition} Let $P_c$ and $P_d$ be shortest paths. Then $c$ swap-needs $d$ if and only if $d$'s goal is on $c$'s path, $c$ and $d$'s paths are parallel, and $d$'s path is a subset of $c$'s i.e. $P_d$ is a subsequence of $P_c$. In notation, $c$ swap-needs $d$ if and only if $g_d\in P_{t,c}$, $P_{t,c} \rightrightarrows P_{t,d}$, and $P_{t,d}\subset P_{t,c}$. 
\end{proposition} %also need to include vacuous alignment, in which d is already at its goal
\par \begin{proof} If those hypotheses hold, then $c$ cannot get to any node on $P_d$ before $d$. In particular, $d$ gets to $g_d$ before $c$. To show the hypotheses are necessary, consider negating each one. We need $g_d$ on $P_c$ to have a swap risk $(c,d)$. Indeed, the only way $c$ can be stuck before reaching its goal is if its next node is currently occupied by $d$, and $d$ does not want to move -- because it is already at $g_d$. If the paths are not parallel, then they are either nonintersecting or else antiparallel. Clearly there is no swap need if the paths are nonintersecting. If they are antiparallel, let their intersection be $v_0,\dots,v_m$, appearing in that ordering in $P_c$, with $m>0$. We can get both robots to goal by advancing (or keeping) $c$ to $v_0$ and advancing (or keeping) $d$ to $v_1$. We can then either immediately swap the two robots, and then send them along their nonintersecting remaining paths. Or else we can move one of the robots off its $v_i$, send the other robot to goal, and then send the first robot to goal. Finally, assume the paths are parallel but there is some node $v$ on $P_d$ that is not on $P_c$. Then $d$ must go through $v$ before stymying $c$ at $g_d$. But then there is no swap need, as $c$ can advance to $g_c$ while $d$ waits at $v$.
\end{proof}
\par This then allows for a characterization of swap risks. Let $P_a = a_0,\dots,a_n$ with $g_b = a_m$ and $P_b$ be shortest paths.% $a$ swap-risks $b$ if and only if $b$ does not swap-need $a_0,\dots,a_{m-1}$. That is
\par \begin{proposition}\label{Swap Risk Characterization} Robot $a$ swap-risks $b$ if and only if either
\begin{enumerate}
\item $P_a \leftrightarrows P_b$, $g_b$ is on $P_a$, and $a$ is not currently occupying $g_b$ i.e. $a_0 \neq g_b$.
\item $P_a\rightrightarrows P_b$, $g_b$ is on $P_a$, and $P_{a,0},\dots,P_{a,m-1}$ is not a subset of $P_b$. Equivalently, $P_{a,0},\dots,P_{a,m}$ is not a subset of $P_b$.
\end{enumerate}
\end{proposition}
\par \begin{proof} If (1) holds, then we cannot have $P_{a,0}$ on $P_b$, because it must come after $g_b$ by antiparallelism. Thus we can incur a swap need $(a,b)$ by advancing $b$ directly to $g_b$. Conversely, if $P_a\leftrightarrows P_b$, then to incur a swap need $(a,b)$, we need that $g_b$ can be occupied by $b$ before $a$: otherwise, if $b$ blocked $a$'s progress, we would always be able to move $b$ or happy swap $a$ and $b$. 
\par $P_{a,0},\dots,P_{a,m-1}$ not being a subset of $P_b$ captures the condition that $b$ can get to $g_b$ before $a$. One way of thinking of this target condition is that, if we truncate $a$'s path to $P_{a,0},\dots,P_{a,m-1}$, then we do not have that $b$ swap-needs $a$. Thus, if (2) holds, then there is some way to finagle $b$ to $g_b$ before $a$, again incurring a swap need. Conversely, we need $g_b$ on $P_a$ to have a swap risk $(a,b)$, and some way for $b$ to get to $g_a$ before $a$ without swapping.
\end{proof}
\par It might seem mysterious that there can be swap risks between robots on antiparallel paths. The explanation is that those paths can become trivially parallel. There is a swap need if and only if that trivial intersection is a goal node $g_a$ occupied by its goalie $a$.
\par It follows from the characterization of swap risk that subset-swapping $a$ past $b$ at time $t$ resolves the swap risk $(a,b)$, as the two new paths are parallel and $P_{t+1,a,0},\dots,g_b = P_{t+1,b,1},\dots,g_b\subset P_{t+1,b}$. Also, it does not introduce a swap risk $(b,a)$, as $g_a$ is not on $P_{t+1,t}$. %and does not introduce a swap risk $(b,a)$. 
The paths are parallel before and after swapping. After swapping, $P_{t+1,a,0},\dots,P_{a,\text{index}(g_b)-1}$ is not a subset of $P_{t+1,b}$, so there is no longer a swap risk $(a,b)$. 
\par \begin{lemma} If $P_c\rightrightarrows P_d$ are shortest paths, then we can have at most one of $g_d\in P_c$ and $g_c\in P_d$.
\end{lemma}
\par \begin{proof} Otherwise we would have $g_d,g_c$ appear in that order in $P_c$ and $g_c,g_d$ appear in that order in $P_d$, and the paths would be antiparallel, contradiction.
\end{proof}
\par We can now return to the proof of the main monovariant lemma.
\par \begin{proof}[Proof of Main Monovariant Lemma] Denote the vertex that $b$ is pulled from and that $a$ swaps to as $v_b$. Similarly, denote the vertex that $a$ leaves and $b$ is pulled to as $v_a$. The only way pulling $b$ back can introduce a swap need between $x$ and $b$ is if $x$ can advance along its path to $v_b$, and in doing so maintain the new swap risk. %That is, it can only introduce a swap need if $x$ first advanced onto $b$'s path. Thus, it must be that $b$ then swap-needs $x$, and that we have $g_x\in P_{t,b}$, $P_{t,b} \rightrightarrows P_{t,x}$, and $P_{t,x}\subset P_{t,b}$. But then, before subset-swapping $a$ and $b$, $a$ must have swap-risked $x$: if $x$ advanced to the same node that made $b$ swap-need $x$, then we would have $g_x\in P_{t,b}$, $P_{t,b} \rightrightarrows P_{t,x}$.
Denote $x$'s shortest path from $v_b$ as $P_{t',x}$, and denote $a$ and $b$'s shortest paths after swapping as $P_{t',a} = P_{t,a}\setminus P_{t,a,0}$ and $P_{t',b}=P_{t,b},P_{t,a,0}$. We will work through three cases. Cases (1) and (3b) cannot give rise to a new swap risk between $x$ and $b$. Cases (2) and (3a) allow for a new swap risk $(b,x)$, but only by resolving the previous swap risk $(a,x)$.
\begin{enumerate}
% \item Assume the swap is a happy swap. Then no new risk can be incurred, by the definition of swap risk. So we can assume that $P_{t,a} \rightrightarrows P_{t,b}$.%$P_{t',a} \leftrightarrows P_{t',b}$. %I think
\item Assume $P_{t',x} \rightleftarrows P_{t',b}$. Since $P_{t',x,0} = P_{t',b,1}$, then we must have that $P_{t',x}$ routes through $v_a$, %by the antiparallel discussion in lemma 2. 
since otherwise they would only intersect in $v_b$ and be parallel. There is no swap risk then: if all other robots are removed, then $b$ and $x$ are either immediately forced to happy swap, or else $x$ takes a detour during which $b$ moves to $v_b$, and then their remaining paths have empty intersection.
\item Assume $P_{t',x} \rightrightarrows P_{t',b}$ and that they only have $v_b$ in common. For there to be a swap risk, we need $v_b = g_x$; otherwise we would resolve the swap risk as soon as we move $x$. So let $v_b = g_x$. That means that there is a swap risk $(a,x)$ at time $t$, as $x$ could advance to $v_b = g_x$ and make $a$ swap-need it. Conversely, advancing $a$ to $v_b$ resolves this swap risk. %If we restrict to $a$ and $x$ and advance $a$ to $v_b$, then we have $g_x\notin P_a\setminus P_{a,0}$, hence $a$ cannot swap-need $x$ anymore, by lemma 4.%we can then advance $x$ to one within $v_
%ok, I'm good if P_a and P_x are parallel, since then a occupies x's goal, and we can't have parallel paths containing each others' goals (and starts). But if they're antiparallel, 
\item Assume $P_{t',x} \rightrightarrows P_{t',b}$ and that they have at least two vertices in common, one of which we know must be $v_b = P_{t',x,0} = P_{t',b,1}$. There are two subcases: $g_x \in P_{t',b}$, and $g_b \in P_{t',x}$, which correspond to swap risks $(b,x)$ and $(x,b)$ respectively. %Ugh, this is a little involved, but it basically amounts to saying that the intersection is contiguous
We use lemma 4.10 to further break into the following two subcases
\begin{enumerate}
\item $g_x\in P_{t',b}$. Then $v_a\notin P_{t',x}$, since, if so, we would have $v_b,v_a$ in that order in $P_{t',x}$, and $v_a,v_b$ in that order in $P_{t',b}$. Thus, while $b$ sits at $v_a$, $x$ can advance from $v_b$ to $g_x$ and incur the swap-need $(b,x)$. This means that there was a swap risk $(a,x)$ at time $t$, as $x$ could similarly have advanced to $v_b$ and then $g_x$ while $a$ sat at $v_a$. Furthermore, we assumed that there was no swap risk $(b,x)$ at time $t$. That is, regardless of how we were to advance $b$ and $x$, $b$ would make it past $g_x$ and eventually to $g_b$. This same reasoning shows that there is no swap risk $(a,x)$ at time $t'$: by parallelism, $g_x$ is the last node from $P_{t',x}$ to appear on $P_{t',a}$; if we restrict to advancing $x$ and $a$ at time $t'$, then we are forced to advance $a$ past $g_x$, and eventually to $g_a$. 
\item $g_b\in P_{t',x}$ and there is a swap risk $(x,b)$ at time $t'$ but not $t$. Let $P_{t,x} = x_0,\dots,x_n$, and let $x_m = g_b$. By the characterization of swap risk, for there to be no swap risk at time $t$ we need $x_0,\dots,x_{m-1}$ to not be a subset of $P_{t,b}$. But then $x_0,\dots,x_{m-1}$ is not a subset of $P_{t',b}$ either: that would only be possible if $x_0 = v_a$, which is not the case. Thus there cannot be a new swap risk $(x,b)$ at time $t'$.%By the same parallelism argument that shows $v_a\notin P_{t',x}$, we in fact more strongly have that $P_{t,x}\setminus P_{t'x}$ avoids $P_{t',b}$ all together. Thus there is no new swap risk: we could just have easily advanced $b$ from $v_b$ to $g_b$ while keeping $x$ at its original location at time $t$.
\end{enumerate}
\end{enumerate}
\end{proof}
We can finally prove the theorem.
\begin{theorem} RIP terminates in $O(k^2+\text{SIC})$ timesteps with every robot at its goal.
\end{theorem}
\begin{proof} By the monovariant lemma, the potential function $\Phi \leq k^2+\text{SIC}$ decreases by at least $1$ every timestep until the robots meet their goals.
\end{proof}

\section{Bubbletree}
\subsection{Comparison with Parallel Sorting}
In addition to the relation to MAPF, there is a similarity between PERR and parallel bubble sorting, especially for dense maps. For the case of linear maps, we can think of the robots as being numbered by the order of their goals, and think of swapping them as swapping two array entries in bubblesort. Of course, a substantial advantage of PERR over sorting with neighbor-swapping operations is that we already know the destination of each entry. %More generally, we can impose a partial order on the robots by restricting the map to a rooted spanning tree.
Consider the case of a binary tree. If we choose the root node in the ``middle'' of the tree such that the ``left'' subtrees and ``right'' subtrees form an equipartition of the root node's children, then, in $n$ steps can we move all robots with left destinations to the left subtrees and vice versa. We can then recursively call this algorithm on the two subtrees in parallel, giving an $O(n)$ makespan, which is best possible. The bubbletree algorithm generalizes to higher degree trees. It works by forming target subtrees for each subtree $T_c$ defined by a child $c$ of the mid node. In the pre-recursive phase, we can again send all robots to their target subtree in $O(n)$ steps. 
\par A disadvantage of PERR over sorting on trees with neighbor-swapping operations is that the graph topology for PERR is specified in advance. We have some freedom in choosing the spanning tree (or forest), but there is much more freedom in building a balanced heap, say, in which the entries need only be partially ordered and in which new nodes can be spontaneously produced.
\par Bubble sorting is a rightly maligned sorting algorithm: its quadratic runtime is much worse than the $n\text{log}n$ bound on comparison sorting achieved by quicksort and mergesort, and, in practice, bubble sort is significantly slower than other quadratic sorting algorithms like insertion sort and selection sort. 
Parallel bubble sorting achieves the maximum possible speed up, proportional to the number of processors if there are at least as many entries as processors (Silva 2016). Even an $O(n)$ runtime is disappointing for parallel sorting, however. There is a known algorithm, which uses the idea of 
expander graphs, that attains the optimal $O(\text{log}n)$ runtime, although the big-$O$ hides prohibitive constants (Ajtai et al. 1984). And 
there are practical parallel sorting algorithms like 
odd-even sort %and %___sort 
that run in $\text{log}^{2}n$ 
time (Haberman 1972). Nevertheless, since robots have to physically move one step at a time, we cannot expect $n$ robots to reach their goals in under $n$ steps. %um, the ' in can't isn't showing up...
Thus parallel bubblesorting is a natural place to look for inspiration and the hope of  optimal big-$O$ results. 
As mentioned above, the analogy between PERR and parallel bubble sorting is still imperfect however, for four reasons:
\begin{enumerate}
	\item We already know the destination of each robot. For example, if we were allowed to perform insertions, we could sort the robots in $n$ steps. This suggests that we might expect to get even better performance in PERR than that obtained by bubblesort-esque comparison sorts. 
    \item PERR maps can have very few robots compared to the total number of nodes.
	\item Parallel bubble sorting is only designed for linear arrays, whereas PERR allows for arbitrary graph topologies.
	\item We care much more about the makespan than the runtime. That is, we do not mind as much about excess comparisons as long as the robots do not have to perform them in the real world.
\end{enumerate}
	The second and third differences will require generalizing parallel bubblesort. The first difference allows us to use specialized heuristics that do not transfer to bubblesorting on general graph topologies. Empirically, RIP is outperformed by bubblesort and shearsort, a cousin of bubblesort specialized for two-dimensional arrays. This is not very surprising given how general-purpose RIP is. We have not yet implemented bubbletree. %which takes closer inspiration from pbs
%Ok, it's pretty unclear how I actually use the ideas from parallel sorting. The answer is that I imagine myself pbs'ing the migrant groups, and I maintain a bound analogous to that in pbs. That's basically it lol.
%Also, bubblesort and shearsort are sorting nets, whereas I can do more dynamic stuff

\section{Bubbletree overview}%and proof of O(diam+k) makespan (er, O(deg*diam+k))
%say somewhere that you maintain the invariant that only robots from $G^T_c$ ever enter $T_c$.
%Recall that $r_m$ only moves when it is blocking an infiltrator chain with target $T_c$, or when all the $G^I_d$ have emptied. If that $G^I_c$ is empty, then $r_m$ runs to $v_m$. If $r_m$ is blocking off space lower down in $T_c$, then $r_m$ buries itself one node deeper.
The idea of bubbletree is to restrict the graph to a tree $T\subset G$ that contains all the start and goal vertices. We can also form a forest, as long as each pair $\{s_i,g_i\}$ of start and goal vertices is contained in the same tree. Throughout, we will give the description of the algorithm for a single tree. Suggestions for how to form the forest are given in the optimizations section.
\par We want to have some guarantee that the robots are making progress. To this end, designate a mid node $v_m$ such that each of the subtrees defined by its children has at most $k/2$ goal vertices: the notion of progress is relative to $v_m$. (One way to find such a mid node is to perform a breadth-first search from an arbitrary vertex $v$, and mark every vertex $v'$ with the number of vertices in the subtree it defines. For example, every leaf $l$'s subtree $T_l$ is size $1$. If $|T_c|\leq|T|/2$ for every child $c$ of $v$, then we are done. Otherwise let $d$ be the unique node for which $|T_d|>|T|/2$. Then $d$ is a satisfactory mid node.) 
% $k/2$ goal vertices. (One way to do this is to perform a breadth-first search from an arbitrary vertex $v$, and mark every vertex $v'$ with the number of goal vertices in the subtree it defines, $g(v')$. For example, $g(v')\leq1$ for leaf vertices. If $g(c)\leq k/2$ for child $c$ of $v$, then we're done. Otherwise let $d$ be the unique goal node for which $g(d)>k/2$. Then $d$ is a satisfactory mid node.) The notion of progress is then relative to this mid node $v_m$.
\par Every robot has an initial subtree: the $T_c$ for the closest child $c$ of $v_m$ (the one irrelevant exception is that a robot $r_i$ with $s_i = v_m$ has no initial subtree.) Every robot $r_i$ also has a target subtree: the $T_c$ for the closest child $c$ of $v$ closest to $g_i$ (there is a special subroutine for handling the robot $r_{m}$ with $g_m = v_m$. If such a robot exists, then it has no target subtree.) We form groups $G^I_c$: those robots with a common initial subtree $T_c$, and groups $G^T_c$: those robots with a common target subtree $T_c$. All robots that are not currently at their target tree are called migrants. We let $M_c\subset G^I_c$ denote the set of robots with initial tree $T_c$ that have not yet made it to their goal tree. The algorithm is designed so that each migrant set $M_c$ weakly shrinks with each timestep. Once every migrant robot has made it to its target tree (and $r_m$ has made it to $v_m$), then we can recursively call bubbletree in parallel on all the $T_c$. 
\par The algorithm works by fixing a priority ordering on each $M_c$. In a given timestep, the main operation is to select a focal $M_c$ and loop through the robots in $M_c$ in priority order, greedily advancing each robot that has not moved yet, even if this requires a bully swap. To select the focal $M_c$, consider the robot $r_j$ that is currently occupying $v_m$ (if no robot is currently occupying $v_m$, the choice of focal $M_c$ is arbitrary.) Let $T_c$ be the target tree of $r_j$. Then $M_c$ is the focal group, and we advance the robots in $M_c$ according to priority order. The idea is that one of those migrant robots will eventually make it to $v_m$, swapping $r_j$ into $T_c$ as desired. Thereafter, $r_j$ will never be displaced from its target subtree $T_c$.
\par If every robot in $G^I_c$ has already gotten to its target tree i.e. $M_c$ is empty, then we instead try to directly funnel the robots in $G^T_c$ (especially $r_j$) into $T_c$. Try to form a chain out of them starting at $v_m$ and terminating at an unoccupied node of $T_c$. If this is possible, then we form such a chain and move them all simultaneously deeper into $T_c$. If we cannot form such a chain, it must be because $r_m$ is in the way. If there are still migrants in $G^T_c$, then we simultaneously move $r_m$ deeper into $T_c$ along with the chain. Else if there are no more migrants in $G^T_c$, then we swap $r_m$ up the chain to $v_m$. Observe that $r_m$ only has to ascend to $v_m$ once per $T_c$. Thus we spend at most $d(s_m,v_m)+2k$ timesteps moving $r_m$: at most $d(s_m,v_m)$ timesteps getting to $v_m$ the first time, at most $\sum_{c}|G^T_c|<k$ steps plumbing deeper into the $T_c$'s, and at most $\sum_c|G^T_c|<k$ steps swapping up along chains of infiltrators. We do not move $r_m$ when recursively calling bubbletree on the subtrees.%Denote the longest path from a leaf of a subtree $T_c$ to $v_m$ by depth$(T_c)$. We must then devote at most depth$(T_c)$ timesteps swapping $r_m$ up to $v_m$. 
% $T_c$ is then completely filled by $G^T_c$, and we can apply bubbleTree to it recursively. Since $r_m$ only has to move this way once per subtree, we spend at most %$\sum_{c \text{ child of }v_m}\text{depth}(T_c)\leq 
% $k$ timesteps moving $r_m$ in this way. %might be depth -1, whatever
% After all migrants have made it to their target trees, we might have to devote another depth$(T^c)$+ moving $r_m$ from whatever $T^c$ it currently occupies.
\par Assume that we have the following invariant: after $t$ timesteps with $M_c$ as the selected group, the $i^{th}$ highest priority robot $r_i$ in $G^I_c$ is guaranteed to have either already made it to its target subtree, or else is at most depth$(T_c)+2i-t$ steps away from $v_m$ (here depth$(T_c)$ denotes the farthest distance a migrant in $M_c$ starts from $v_m$. Also, the highest priority robot has priority $i=0$.) If it is not the turn to move $M_c$, then this invariant is maintained. Indeed, robots are never pulled back across $v_m$. And if a migrant from $G^I_c$ that has not gotten to $v_m$ yet is swapped, then that swap will only serve to bring it closer to $v_m$. If it is the turn to move $M_c$, then consider what it means if some migrant $r_i\in M_c$ that has not yet gotten to $v_m$ fails to advance. There are two possibilities. The first is if a higher priority robot $r_j$ swaps another robot to the node that $r_i$ was planning on advancing to (this is called snubbing.) Then $r_i$ is only one step worse off than $r_j$ was the turn before at $t-1$. The guaranteed proximity to $v_m$ for $r_j$ was depth$(T_c)+2(i-1)-(t-1)=\text{depth}(T_c)+2i-t-1$. So the guarantee for $r_i$ is  $\text{depth}(T_c)+2i-t$ as desired. The second is that a higher priority robot $r_j$ bully-swapped it. This bodes even better for $r_i$ than if it is snubbed, and the guarantee closeness to $v_m$ is again maintained. %But then $r_i$ has only a one-point-worse guarantee, depth$(T^I_c)+2(i-2)-(t-1))$, than $r_j$ had the turn before (depth$(T^I_c)+2(i-1)-t)$). We can check that depth$(T^I_c)+2(i-2)-t-1) = \text{depth}(T^I_c)+2(i-2)-t-1)$
The base case holds for $t=0$, as every robot in $M_c$ is at most $\text{depth}(T_c)$ away from $v_m$, and the non-migrants in $G^I_c$ start out in their target tree $T_c$. % or $v_m$.%, and every robot in $G^I_c\setminus M_c$ is either at $v_m$ or has already infiltrated its target tree.
So we are done with the pre-recursive phase after at most $\sum_{c \text{ child of }v_m}\text{depth}(T_c)+2|M_c|$ timesteps advancing the various $M_c$'s, plus the timesteps during which $r_m$ moves.
% \par For bounded-degree graphs like grid graphs, we conclude that the makespan is bounded by 
\subsection{Basic Bubbletree Makespan Analysis}
\par %The basic bubbletree algorithm performs well on bounded-degree graphs with about as many robots as nodes in which some robot needs to traverse nearly the entire graph to reach its goal. 
Imagine we have already formed a bubble tree. Denote the max degree by $d_{deg}$ and the greatest depth by $d_{depth}$%longest path length by $l$.
Assume that $k = \theta(n)$, where $n$ is the number of vertices. The pre-recursive phase then takes at most %$dl+4k = O(n)$ 
$d_{deg}d_{depth}+4k = O(n)$ steps ($2k$ comes from the $2|M_c|$ terms, and another $2k$ comes from moving $r_m$.) Let $f(d,n)$ denote the maximum makespan produced by the basic bubbletree algorithm for a graph with max-degree $d_{deg}$ and and $n$ nodes. We then have the recursion $f(d_{deg},n) \leq (d_{deg}d_{depth}+4k)+f(d,n/2) \leq (d_{deg}n+4n)+f(d_{deg},n/2)$. It can easily be verified that $g(d,n/2) = 2dn+8n$ satisfies the recurrence $g(d,n) = (dn+4n)+g(d,n/2)$. Thus $f(d_{deg},n) \leq 2d_{deg}n+8n$.
\par Of course, this result is not very impressive if there are not very many robots, or if the max-degree is unbounded.  %riiiiight. So the optimizations are: goal-balance rather than node-balance; head start infiltators; move all M_c in parallel.
In the optimizations section we present two modifications that give a provable makespan of $O(\text{diam}*\text{log}k+k)$, where diam is the diameter of the bubbletree. % (slightly more strongly, the makespan is $O(l*\text{log}k+k)$, where $l$ is the longest distance from start to goal in the bubbletree.)
They are 
\begin{enumerate}
\item Choosing $v_m$ so that its subtrees are goal-balanced, rather than node-balanced.
\item Moving as many $M_c$ as possible in a given timestep, rather than just the focal $M_c$.
%\item Devoting some time at the start of each recursive call of bubbletree to bring all robots to within $k'$ of $v_m'$ (here $k'$ and $v_m'$ denote the number of robots and mid node in the recursive call.) This way we effectiveley have $l'\leq k'$.
\end{enumerate}
\par We would attain an appealing linear result without the $\text{log}k$ factors. Unfortunately, no algorithm that prunes down to an acyclic graph can achieve an $O(l+k)$ makespan in all cases, where $l$ is the max distance from start to goal in the entire graph. Consider $n/\text{log}n$ robots on a circular graph with $(s_i,g_i) = (i\text{log}n,(i+1)\text{log}n+1)\pmod{n}$. The union of the shortest paths is the entire graph. Thus we must destroy some robot's shortest path if we subtract an edge. The optimal makespan is $\text{log}n+1$, obtained by advancing all the robots simultaneously in the circle. However, the optimal makespan with an edge deleted is at least $n-(\text{log}n+1) = \Omega(n) = \omega(k+1)$.%\omega(\text{log}n*= \Omega(l*\text{log}k+k)$. 
\par There may still be hope of a linear result in terms of diam or even $l$ for a tree graph ($l$ denotes the longest distance between some $s_i$ and $g_i$.) For example, if we could guarantee that the subtrees we recursively call bubbletree on had either diameters or longest paths that shrank by a constant factor, then we could prove a linear result (assuming this lengthens the pre-recursive phase by only a constant factor). 
\par However, we cannot guarantee shrinking diameter if one of the subtrees is linear. We cannot even guarantee constantly shrinking diameter after many steps, as the tree could have a linear backbone with highly-branched offshoots, say with $2^{i-1}$ nodes at distance $1$ from the $(2^i)^{th}$ node of the backbone (or at distance $\text{log}(2^{i-1})$, if the tree has bounded degree). We cannot guarantee constantly shrinking maximum path length, either. %, as I will present after exploring in more depth in the optimizations section.
For example, we could again have a tree with a long backbone and exponentially shrinking offshoots. Say there is an offshoot of $2^{\text{lg}k-i}$ goal nodes emanating from the $i^{th}$ node of the backbone, and say the backbone itself contains all the start nodes and is very long. If the tail of the backbone is comprised of start nodes from the $\text{lg}k$ different offshoots, then there will be a recursive call at all depths $1,\dots,\text{lg}k$ for which the maximum length to goal will decrease by at most $\text{lg}k$. If the backbone is sufficiently long, it will also be the case the $\text{diam}*\text{lg}k$ dominates $k$. The optimization of advancing infiltrators (which does not wait for recursive calls in order to advance the robots starting at the tail of the backbone) effectively reduces $\text{diam}$ to $\text{lg}k$ in time for the first recursive call in this example. But we do not know if it is strong enough to give a linear makespan in general.

\section{Basic Bubbletree pseudocode}

\begin{algorithm}[h]
\small
\caption{Basic Bubbletree algorithm High-Level}
\label{BubbleTree algorithm High-Level} %%?
%%%ok, I don't know how to write pseudocode comments, so I'll mark them as %//
	%// form tree T\subset G that contains all the start and goal vertices
    $T \gets formBubbleTree(G)$\;
    $v_m \gets findMidVertex(T)$\;
         	\For{each child $c$ of $v_m$}
  	{form the migrant set $M_c$ of robots in $T_c$ that are not in their target tree\;
    }
     	\For{each child $c$ of $v_m$}
  	{ assign a (random) ordering to the robots in $M_c$\;
    }
    \While{not all $M_c$ are empty}
    {
    \If{$v_m$ is occupied, by say $r\in G^T_c$} 
     {
     \If{$focalM$ is nonempty}
     {$focalM \gets M_c$\;
     priorityMove($M_c$)\;}
         \Else{
     {moveChain($G^T_c$)}
     }
     }
         \Else{
     {$focalM \gets randomNonemptyM()$\;
                priorityMove($M_c$)\;}
     }
 %      \If{$focalM$ is nonempty} {
    %       priorityMove($M_c$)\;
       %    }

     }

 	moveRmtoVm()\;
    \For{each depth-1 subtree $T_c$ in parallel} {
    bubbleTree($T_c$)\;
    }

\end{algorithm}

\section{Experimental Results}
%Comparisons with RIP against the optimal solvers from Hang's paper.
%Comparisons with RIP against bubblesort and shearsort.
%I don't have any results for bubbletree, sorry.
We tested RIP on: the Dragon Age Origins map brc2d (Sturtevant 2012) with up to $50$ robots; on 20x15 grid maps with randomly generated obstacles, and under a quarter as many robots as free nodes; and on linear and square arrays with as many robots as nodes. For both the DAO and 20x15 grid maps, we used the exact same (randomly generated) start and goal positions as in (Ma 2016). We compare our results against the optimal flow-based and cbs-based solvers used in (Ma 2016). %We use as many robots as nodes in the randomly generated linear and square arrays. 
We compare the linear and square array results against parallel bubblesort and shearsort, which are specialized for sorting linear and 2D arrays, respectively.
\subsection{DAO Map}
Adaptive CBS is an optimal solver, but somewhat slower than the flow-based solver introduced in (Ma 2016). Table \ref{DAO table} thus demonstrates that RIP performs near optimally on the DAO tests. Note that the ``super-optimal'' RIP makespans are due to the fact that many of the instances that Adaptive CBS did not solve in the allotted time have smaller than average makespans. In our tests, the DAO map is robot-sparse: it has $254,930$ cells, and at most $50$ robots. These tests corroborate the intuition that RIP performs well on maps with relatively few conflicts.
\par The runtime scaling for RIP appears to be linear with the number of robots, which also makes sense given the relative paucity of collisions: it is the case that almost every step in the inner while loop should advance a robot. For contrast, the runtime scaling for Adaptive CBS is more dramatic -- $63\%$ of all tests with $50$ robots did not even finish in the allotted five minutes. For Adaptive CBS, the median runtime is over three times that of the average completed runtime, suggesting the distribution of runtimes for a map with randomly seeded starts and goals is heavily-tailed as well.

\begin{table}[h]
\centering
\caption{RIP vs. Adaptive CBS on DAO map brc202d}
\label{my-label}
\begin{tabular}{|l|l|l|l|l|l|}
\hline
Adaptive CBS &                 &          &        & RIP      &         \\
\hline
Robots       & Fraction Solved & Makespan & Time   & Makespan & Runtime \\
\hline
5            & 1.00            & 732.1    & 0.34s  & 732.13   & 1.23s   \\
10           & 1.00            & 809.03   & 17.75s & 809.17   & 1.40s   \\
15           & 0.97            & 882.28   & 3.51s  & 880.90   & 2.1s    \\
20           & 0.87            & 905.15   & 5.43s  & 908.03   & 2.63s   \\
25           & 0.97            & 931.34   & 22.73s & 937.40   & 2.8s    \\
30           & 0.87            & 942.19   & 29.03s & 938.10   & 3.77s   \\
35           & 0.77            & 963.13   & 50.80s & 969.33   & 4.67s   \\
40           & 0.53            & 974.25   & 30.50s & 965.80   & 5.67s   \\
45           & 0.70            & 974.1    & 77.49s & 962.87   & 5.77s   \\
50           & 0.37            & 943.36   & 86.76s & 962.43   & 6.43s  
\\ \hline
\end{tabular}
\label{DAO table}
\end{table}

\subsection{20x15 Grids}
Here density denotes the fraction of nodes that were replaced with impassable obstacles. The flow-based solvers are presented in (Ma 2016). The suboptimal flow solver works by first restricting the graph to the union of every robot's individual shortest path (one path per robot, as in RIP), and then performing optimal-flow on the pruned graph. 
\par The makespans from RIP are competitive with the optimal makespans on the easier problems, and the RIP makespan is never more than $30\%$ worse than the optimal makespan. Importantly, the runtime for RIP appears linear with the number of robots and invariant of the number of obstacles, whereas the flow-based solvers demonstrate their non-polynomial nature on the hardest instances.
\par With the tremendous focus on swaps in the proof for RIP, we might hope that RIP performs very few swaps. However, RIP makes more swaps than the optimal flow algorithm in all instances, even though the flow solution does not try to minimize the number of swaps. RIP performs about as many swaps as suboptimal flow: it performs slightly fewer swaps in the instances with fewer robots, and slightly more in the instances with more robots.

\begin{table}[!htbp]
\centering
\caption{20x15 Grid Map Data}
\label{20x15 Grid Map Data}
\scalebox{0.75}{
\begin{tabular}{|l|l|l|l|l|l|l|l|l|l|l|}\hline
					  &         & RIP      &         &           & Optimal Flow &         &           & Suboptimal Flow &         &           \\
\hline
Robots                & Density & Makespan & Runtime & Swaps	 & Makespan     & Runtime & Swaps 	& Makespan        & Runtime & Swaps \\
\hline
10                    & 0\%     & 20.73    & 0.3s    & 1.83      & 20.67        & 1.78s   & 0.70      & 20.67           & 0.55s   & 2.70      \\
10                    & 5\%     & 21.00    & 0.03s   & 1.97      & 20.77        & 1.61s   & 1.07      & 20.77           & 0.53s   & 2.37      \\
10                    & 10\%    & 21.63    & 0s      & 2.47      & 21.13        & 1.52s   & 0.83      & 21.13           & 0.54s   & 3.13      \\
10                    & 15\%    & 22.60    & 0s      & 2.27      & 22.37        & 1.65s   & 1.57      & 22.37           & 0.64s   & 3.30      \\
10                    & 20\%    & 23.33    & 0.03s   & 2.93      & 23.23        & 1.57s   & 2.23      & 23.23           & 0.63s   & 3.77      \\
10                    & 25\%    & 24.80    & 0.03s   & 3.63      & 24.70        & 1.52s   & 3.50      & 24.70           & 0.70s   & 4.97      \\
10                    & 30\%    & 31.43    & 0.03s   & 6.77      & 30.67        & 2.01s   & 5.37      & 30.73           & 1.02s   & 7.03      \\
\hline
20                    & 0\%     & 24.00    & 0.5s    & 8.17      & 23.67        & 5.20s   & 6.83      & 23.67           & 2.05s   & 9.87      \\
20                    & 5\%     & 24.37    & 0.17s   & 7.93      & 23.77        & 4.89s   & 5.77      & 23.77           & 2.09s   & 10.90     \\
20                    & 10\%    & 24.93    & 0.07s   & 9.30      & 24.07        & 4.61s   & 6.50      & 24.07           & 2.19s   & 11.67     \\
20                    & 15\%    & 24.23    & 0.03s   & 10.40     & 23.67        & 3.92s   & 7.80      & 23.67           & 1.87s   & 12.77     \\
20                    & 20\%    & 26.23    & 0.07s   & 13.27     & 25.20        & 3.88s   & 9.37      & 25.20           & 2.00s   & 15.20     \\
20                    & 25\%    & 29.77    & 0.23s   & 17.80     & 28.77        & 5.05s   & 13.27     & 28.77           & 2.95s   & 19.97     \\
20                    & 30\%    & 34.50    & 0.1s    & 24.13     & 33.03        & 7.20s   & 20.80     & 33.03           & 3.97s   & 26.30     \\
\hline
30                    & 0\%     & 25.00    & 0.2s    & 16.30     & 24.10        & 9.52s   & 13.70     & 24.10           & 5.41s   & 18.40     \\
30                    & 5\%     & 26.33    & 0.07s   & 18.50     & 25.63        & 10.81s  & 16.80     & 25.63           & 6.48s   & 22.10     \\
30                    & 10\%    & 25.33    & 0.1s    & 20.17     & 24.20        & 8.59s   & 15.30     & 24.20           & 4.73s   & 21.37     \\
30                    & 15\%    & 26.23    & 0.13s   & 23.97     & 24.97        & 8.07s   & 19.00     & 24.97           & 4.71s   & 24.47     \\
30                    & 20\%    & 28.27    & 0.23s   & 33.07     & 26.70        & 9.17s   & 23.60     & 26.77           & 6.91s   & 31.27     \\
30                    & 25\%    & 33.23    & 0.07s   & 43.60     & 30.27        & 22.16s  & 30.37     & 30.27           & 22.90s  & 38.73     \\
30                    & 30\%    & 37.77    & 0.17s   & 54.77     & 34.87        & 36.47s  & 41.53     & 34.87           & 35.46s  & 53.63     \\
\hline
40                    & 0\%     & 26.43    & 0.13s   & 30.20     & 25.13        & 23.78s  & 23.83     & 25.13           & 13.31s  & 31.93     \\
40                    & 5\%     & 27.33    & 0.1s    & 34.50     & 25.67        & 22.55s  & 27.77     & 25.67           & 15.09s  & 33.60     \\
40                    & 10\%    & 27.50    & 0.07s   & 39.00     & 25.40        & 20.76s  & 31.67     & 25.43           & 14.25s  & 38.53     \\
40                    & 15\%    & 28.37    & 0.2s    & 42.20     & 25.67        & 18.75s  & 34.83     & 25.67           & 16.35s  & 41.77     \\
40                    & 20\%    & 29.73    & 0.33s   & 55.03     & 25.33        & 22.27s  & 39.53     & 25.33           & 19.67s  & 44.70     \\
40                    & 25\%    & 32.20    & 0.13s   & 66.73     & 28.77        & 53.38s  & 46.20     & 28.77           & 61.20s  & 55.40     \\
40                    & 30\%    & 39.73    & 0.2s    & 92.10     & 33.20        & 130.33s & 68.83     & 33.20           & 92.78s  & 79.53     \\
\hline
50                    & 0\%     & 28.23    & 0.3s    & 49.57     & 26.20        & 78.38s  & 39.93     & 26.20           & 55.31s  & 46.50     \\
50                    & 5\%     & 29.20    & 0.3s    & 53.03     & 26.07        & 61.78s  & 40.93     & 26.07           & 49.57s  & 51.73     \\
50                    & 10\%    & 28.40    & 0.2s    & 59.00     & 25.40        & 46.91s  & 45.50     & 25.40           & 37.02s  & 54.00     \\
50                    & 15\%    & 30.47    & 0.23s   & 71.60     & 26.73        & 68.54s  & 55.43     & 26.73           & 56.52s  & 62.20     \\
50                    & 20\%    & 31.73    & 0.33s   & 79.67     & 27.50        & 89.05s  & 61.17     & 27.50           & 94.40s  & 67.53     \\
50                    & 25\%    & 37.70    & 0.4s    & 109.50    & 30.33        & 259.08s & 75.50     & 30.37           & 268.41s & 86.83     \\
50                    & 30\%    & 44.93    & 0.33s   & 154.87    & 34.97        & 596.25s & 105.53    & 34.97           & 522.42s & 123.83   
\\ \hline
\end{tabular}
}
\end{table}

\subsection{Linear Arrays}
We test on randomly generated linear arrays of length $n=k$ for $n=1,\dots,1000$, and run ten tests for each length. The makespans produced by RIP are remarkably linear, and on average slightly less than the length of the array $n$. Curiously, RIP often produces makespans greater than $n$. For contrast, the sorting network implementation of parallel bubblesort given at http://tinyurl.com/damu92 is designed to run in exactly $2n-2$ timesteps. In fact, the algorithm can be terminated after only $n$ timesteps. To see why, recall that serial bubble sort works by launching bubbles from the first array entry so that they catch the largest entry and bring it to the end of the array, stopping before the numbers that have been brought to the end by earlier bubbles. Consider instead launching bubbles as soon as possible, rather than waiting for $n$ (or $n-i$) timesteps as in serial bubble sort. An indispensable graphic is given on slide $9$ of http://tinyurl.com/jegusez. Assume the length of the array is even. We can imagine launching two new bubbles from either end of the array every two timesteps. For $i=1,\dots,n/2$, the bubbles launched on timestep $2i-1$ eventually catch robots $i$ and $n+1-i$.%, starting with robots $1$ and $n$.
 After catching their robots, those bubbles do not release them until they make it to positions $n+1-i$ and $i$. Those robots have then made it to their goals, and will not be disturbed. The $i^{th}$ pair of bubbles finishes at timestep $(2i-1)+(((n+1-i)-i)-1) = n-1$. Accounting for the possibility that $n$ may be odd gives the bound of $n$ timesteps as desired. Since we care about the robots getting all the way to their goals, rather than to the mid vertex, we think the bubbles give a nicer illustration than the case-based guarantees used in the proof of bubbletree. %There is also presumably an easy proof using the 0-1 law. A good explanation of the 0-1 law, and an application to shearsort, is given at %Yeah, bubblesort is so shitty that I can't find A SINGLE PROOF of this, or even a single cited paper mentioning it. Anyway, http://www.personal.kent.edu/~rmuhamma/Algorithms/MyAlgorithms/Sorting/bubbleSort.htm has the pseudocode at least. 
%UGH, anyway, I can think of two ways to prove that parallel bubblesort runs in $2n-2$ steps. The first way 
\par It is also clear that at least $n$ steps are needed in the worst case. For example, if $n$ is odd and the two robots at the endpoints of the array both need to traverse the entire array to get to their goals.

\begin{center}%[!htbp]
\includegraphics[height=95pt]{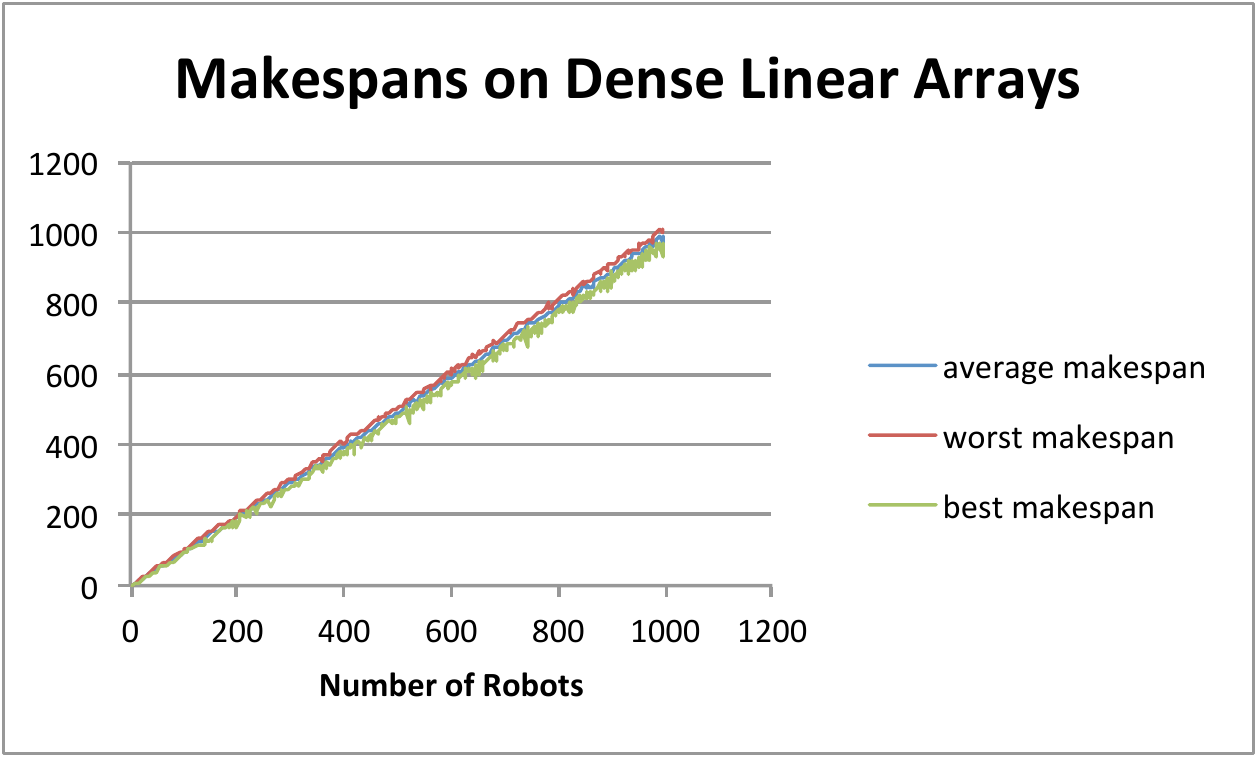}
 \captionof{figure}{Best, worst, and average makespans for the dense linear array tests.}
\end{center}

\par The runtime per timestep appears to be quadratic. %Indeed, it is unsurprising if the runtime per move is linear, as 
In the worst case, we would expect a quadratic number of seconds per timestep -- this could come from any of the inner while loop, resolving cycles, and swapping. Indeed
\begin{enumerate}
\item We could move $\theta(k)$ robots, and have to loop through all $k$ each time. This is improved to linear time in the optimizations section using inverse chains. (Of course, with $k=n$, we only ever resolve cycles and perform swaps.)
\item We could build a length $\theta(k)$ chain for each robot, only to have it terminate at a robot that has already moved or is at its goal. This is also improved to linear time using inverse chains.
\item It can take $\theta(n)$ steps to check if two paths are subsets. For example, if both paths are length $\theta(n)$ and only differ in one node. Presumably we could record swap risks and check when they might become swap needs, but we do not know a good way to do this.
\end{enumerate}

\begin{center}%[!htbp]
\includegraphics[height=\graphheight pt]{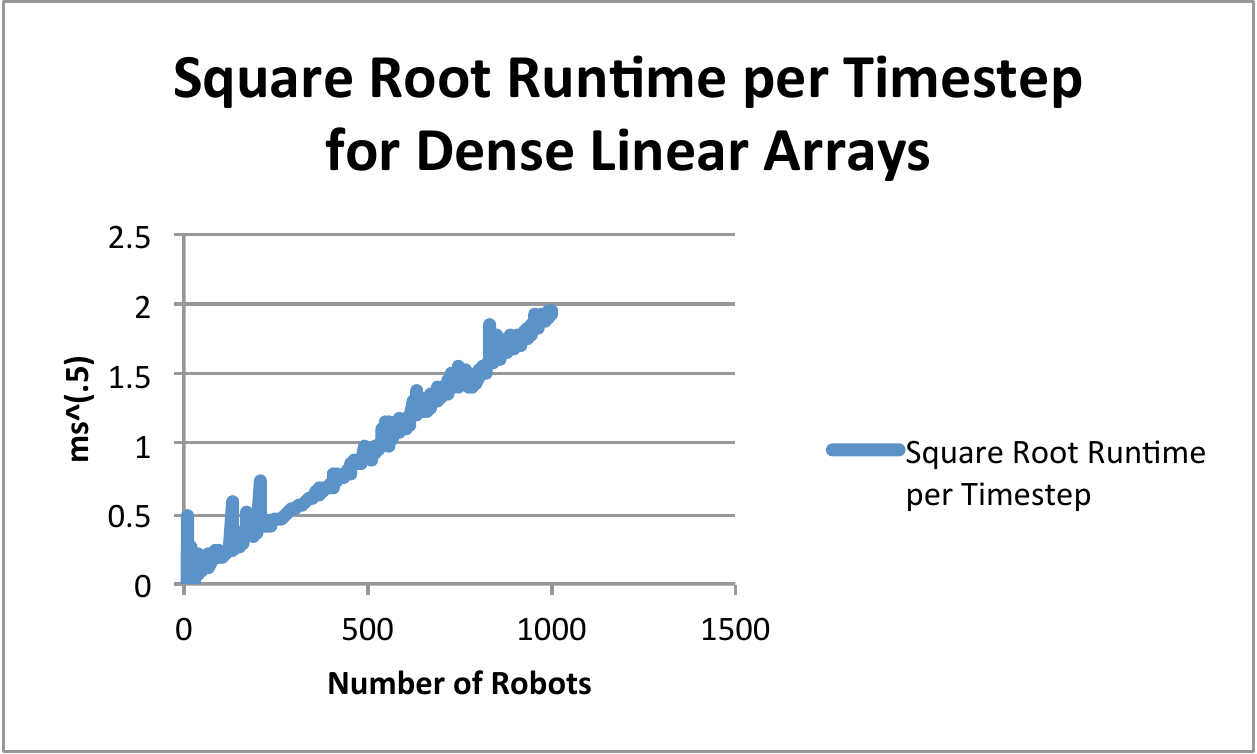}
 \captionof{figure}{Square root of the average runtime per timestep for the tests on the dense linear arrays. Note that there are several small spikes present due to the fact that my computer was running other programs during the tests.}
\end{center}

% \begin{center}%[!htbp]
% % \includegraphics[height=95pt]{tug1}\includegraphics[height=95pt]{RIP vs Shearsort}
% \includegraphics[height=95pt]{DAO_maps.png}\includegraphics[height=95pt]{DAO_maps.png}
%  \captionof{figure}{Self-driving aircraft towing vehicles borrowed from~\protect\cite{airporttug} (left); Graphical representation of surface of part of Dallas Fort Worth Airport borrowed from~\protect\cite{airporttug16} (right).}
% \end{center}

\subsection{Square Arrays}
We test on randomly generated square arrays of length $\sqrt{n}=\sqrt{k}$ for $\sqrt{n}=1,\dots,50$, and run ten tests for each length. Superficially, the makespan graphs appear linear with the number of robots. However, RIP produces makespans that are consistently greater than the number of robots. In fact, inspecting the data suggests the makespans are superlinear. The makespan divided by the number of robots produces a nonconstant, and possibly logarithmic, curve. This is surprising for two reasons:
\begin{enumerate}
\item We could also form a hamiltonian path on the square grid. From the tests on linear arrays, we are convinced that this would produce a makespan slightly less than the number of robots.
\item Shearsort sorts square arrays in 
$2\text{lg}n\sqrt{n}$ serial timesteps -- much faster than for linear arrays.
\end{enumerate}
More precisely, shearsort takes $2\lceil \text{lg}n\rceil$ times the amount of time it takes to sort a single row, which is $\sqrt{n}$ timesteps with parallel bubblesort. For the sake of simplicity, we ignore the ceiling in my graphs.
See http://tinyurl.com/zue7shp for a brief description and proof of completeness for shearsort.

\begin{center}
\includegraphics[height=95pt]{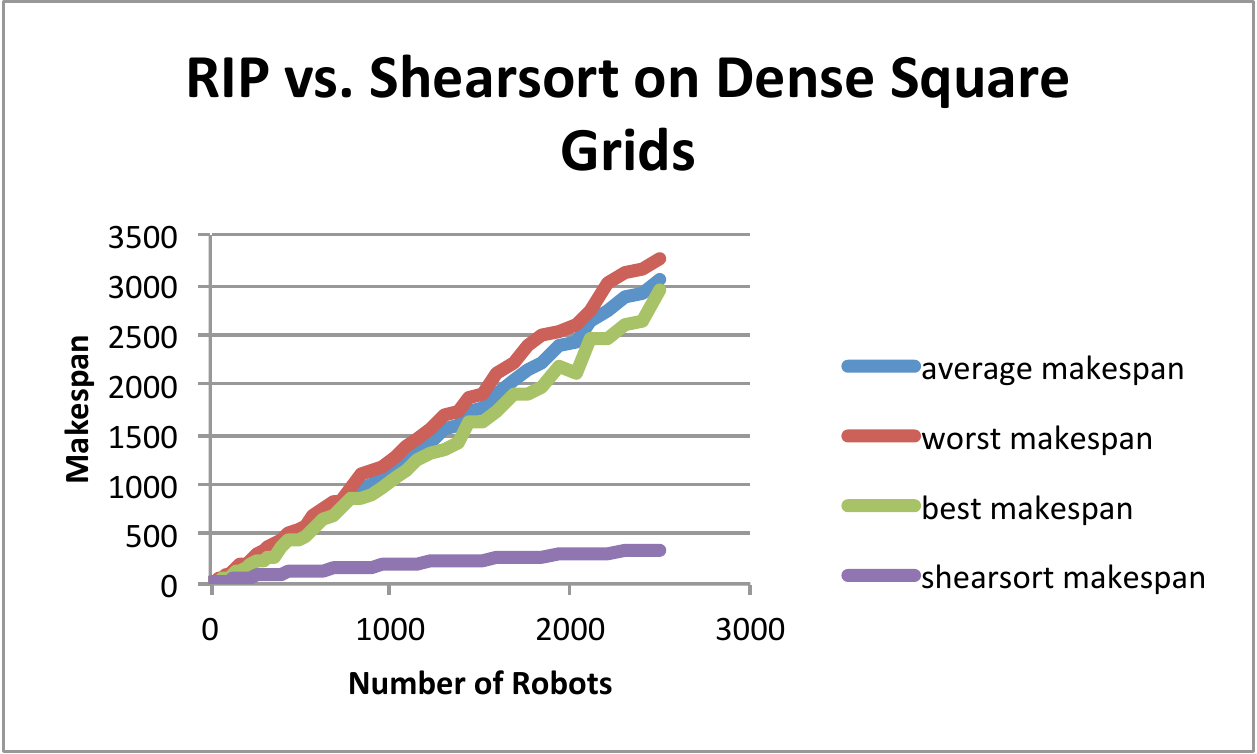}
\includegraphics[height=95pt]{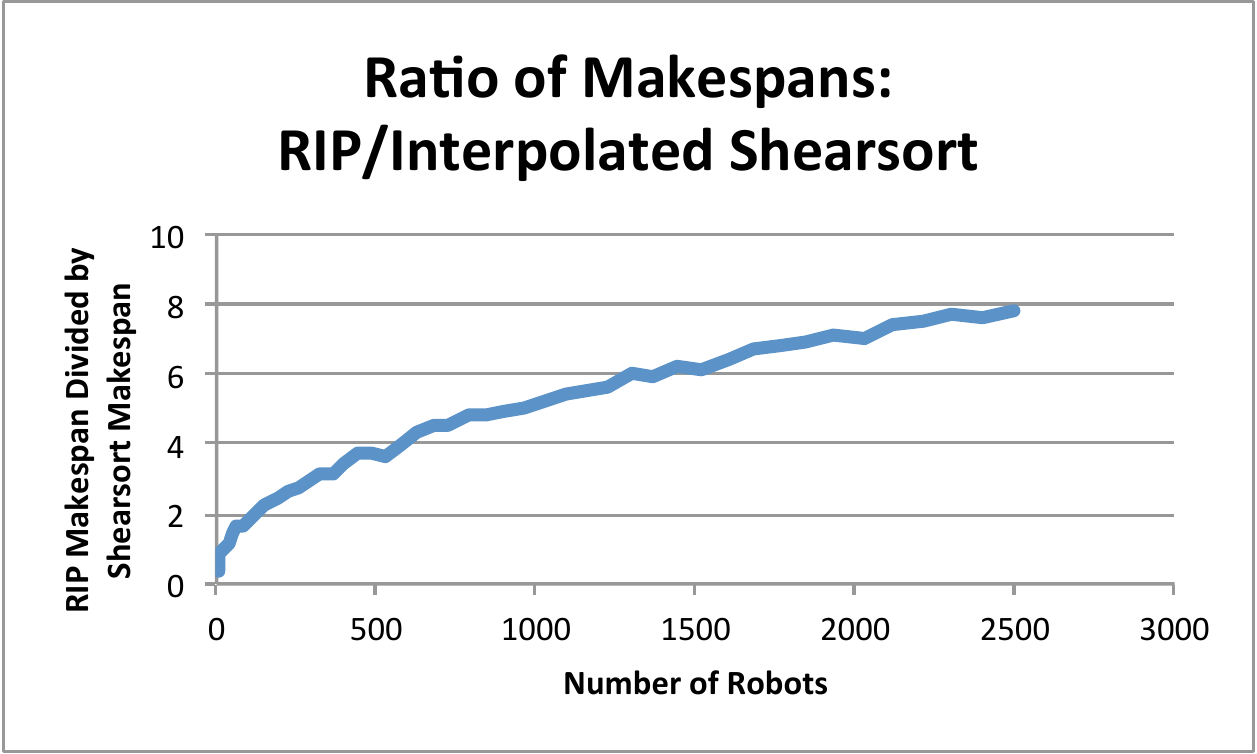}
\includegraphics[height=95pt]{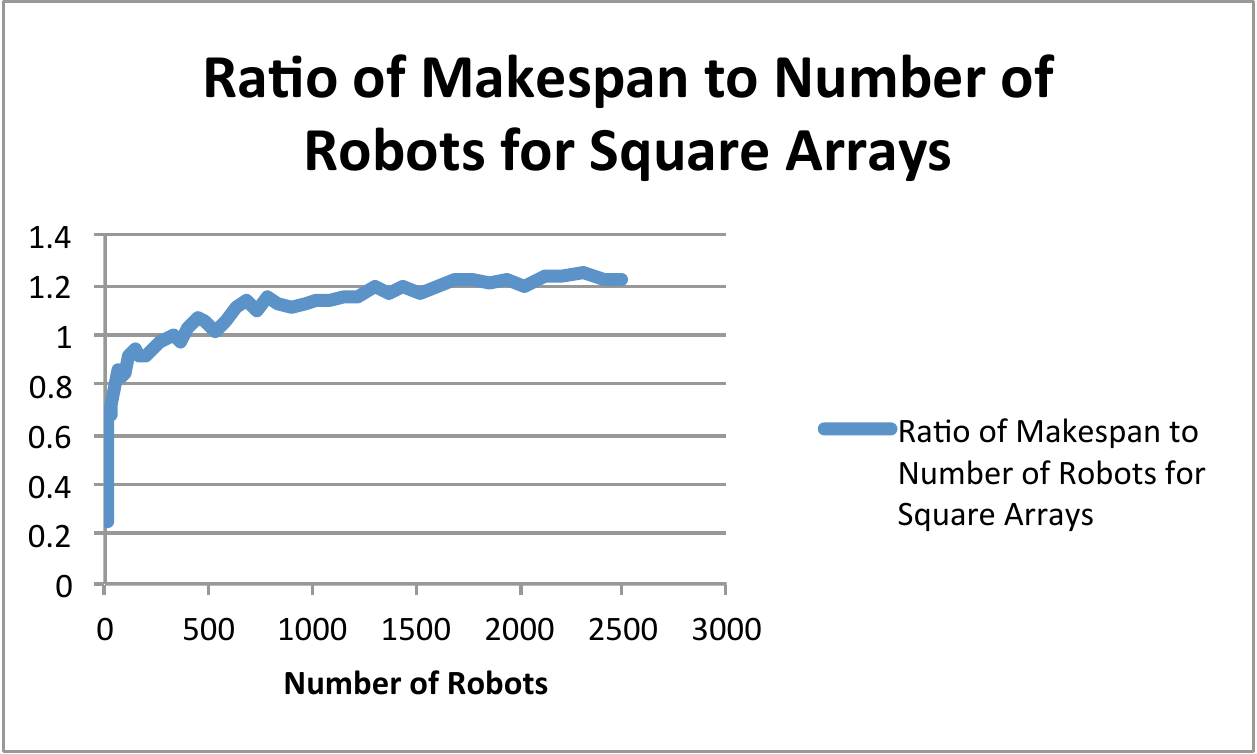}
 \captionof{figure}{The first two graphs are a comparison between RIP and shearsort. The third graph shows that the makespan produced by RIP grows superlinearly.}
\end{center}

\par The (square root) runtime per timestep is almost exactly the same as in the linear array tests. The square array tests are about $10\%$ faster per timestep at the $1000$-robot mark. This is directly offset by the fact that the square array makespans are about $10\%$ larger at the $1000$-robot mark.

\begin{center}
\includegraphics[height=95pt]{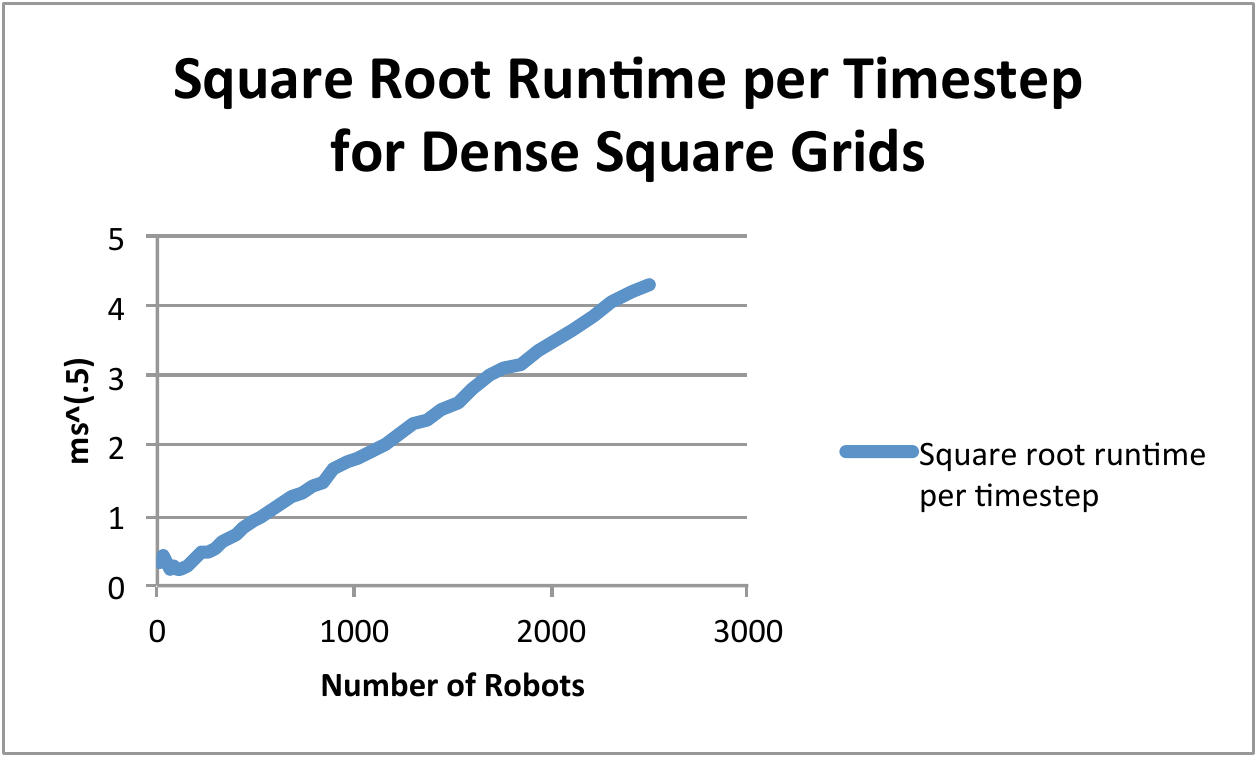}
 \captionof{figure}{Square root of the average runtime per timestep for the tests on the dense square arrays.}
\end{center}

\par The number of swaps also appear super linear. However, it grows much more slowly than $\binom{n}{2}/2$ which is the expected number necessary swaps for dense linear arrays (every pair of robots has a $1/2$ chance of being in the wrong order.)

\begin{center}%[!htbp]
\includegraphics[height=95pt]{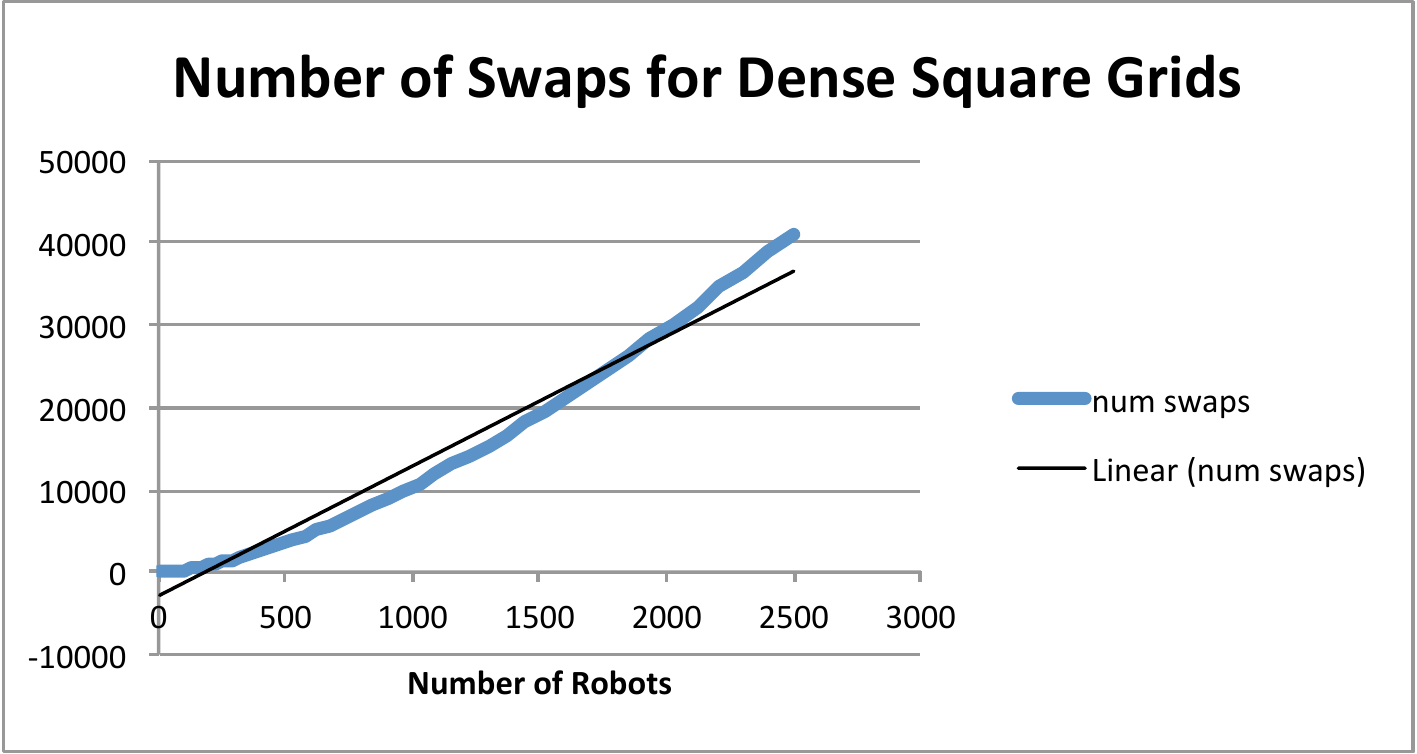}
\includegraphics[height=95pt]{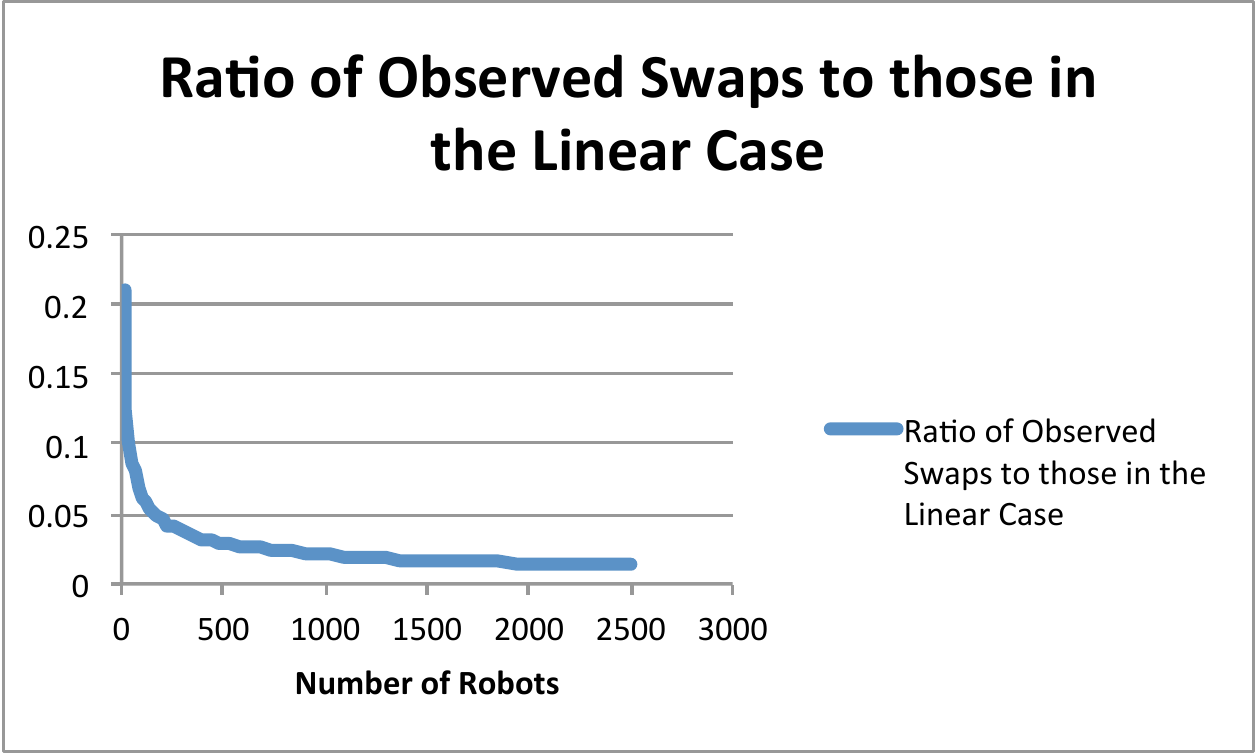}

 \captionof{figure}{The number of swaps appear to grow superlinearly but subquadratically.}
\end{center}

%http://www.iti.fh-flensburg.de/lang/algorithmen/sortieren/twodim/shear/shearsorten.htm

\par It should also be noted that I had to manually remove fewer than ten outliers total for the average runtimes for the linear and square tests. I was using my computer while running the tests, and there were several times when the test runtime exploded by several orders of magnitude while my computer processed a burdensome but unrelated task.

\section{Optimizations}
\subsection{RIP Optimizations}

\subsubsection{Inverse Chains}
% Rather than continuously looping through the robots to see if one has opened up a spot by moving, 
We can perform all the advances in the inner while loop of RIP in a single pass, and we can combine it with the cycle phase as well. 
%To do this, we can create a map from each 
The idea is to partition the robots into doubly linked maps from each robot to the robot occupying the node it wants to go to and vice versa. We can get the exact same behavior as in RIP as follows: Consider a single such map. Form a maximal list by starting at a random robot and going as far in the forward direction as possible. If going forward this way ends with a robot at its goal, then go on to the next inverse chain; that goalie robot will be swapped in the subset-swapping phase (or else that robot is the entire inverse chain). If going forward this way forms a cycle, then resolve the cycle. If going forward this way terminates at an empty node, then build a maximal chain by working backwards from the empty node, always selecting the robot whose index is the least greatest than the current robot's index, mod $k$ (the empty node is treated as index $0$.) Advancing that chain removes the need to loop. %as possible along both the next and previous directions until reaching a null neighbor or else forming a cycle. 
All told this only take $O(k)$ operations per timestep. For contrast, the original inner while loop in the basic implementation of RIP takes $\theta(k^2)$ runtime steps in the worst case, and checking for cycles can take $\theta(k^2)$ runtime steps as well (for example, if the inverse chain is linear).
%We can then either advance the chain or resolve the cycle. In both cases mark that map as moved and go on to the next map. All told this only take $O(k)$ operations per timestep, rather than the worst case of $\theta(k^2)$ for the original inner while loop in the basic implementation of RIP.
\par We can also use inverse chains to get makespan optimizations. If there are no cycles in an inverse chain, then we potentially have many maximal chains to choose from. Intuitively, it makes sense to advance the longest chain. Another idea is to advance the chain with the robot that still has the farthest to travel. We can also build out a chain by starting at the empty node and, at each choice when working backwards, choosing the robot that still has the farthest to go. Choosing robots this way can have the advantage of eliminating swap risks. There are some maps in which it can cut the makespan in half. 
\par For example, in the basic implementation of RIP, say all the robots are vying for the same node $v$, and the rest of their paths are linearly nested, with $P_{i,1},\dots,P_{i,\text{end}}\subset P_{i+1,1},\dots,P_{i+1,\text{end}}$. The proof of completeness did not depend on the order of the inner while loop, the subset-swaps, and resolving cycles. Imagine that we perform the subset-swaps before the inner while loop. How would this play out? First $r_1$ would enter, only to be subset-swapped by $r_2$ in timestep $2$, which would be subset-swapped by $r_3$ etc. Finally $r_k$ would enter $v$ at time $k$, and then go off to $g_k$. Meanwhile, the entire process would have to start over again until $r_{k-1}$ could enter $v$. Overall, it would take $k(k-1)/2$ timesteps just to resolve the massive bottleneck at $v$. Indeed, if the paths are made as short as possible, then this family of examples gives a $\theta(\text{SIC}+k^2)$ makespan with basic RIP. Of course, we could resolve it in $k$ steps by using any of the various inverse chain optimizations. If we perform the inner while loop before the subset-swaps, then small tests produce makespans between $2k$ and $3k$. So, compared to the baseline implementation of RIP, inverse chains do not give a dramatic improvement on small tests.
\par We have only implemented the basic version of RIP, however, and cannot compare what effect these optimizations would have on larger and more typical grid maps.

\begin{algorithm}[!htbp]
\small
\caption{RIP with Inverse Chains}
\label{RIP with Inverse Chains} %%?
    \For{i=1,\dots,k}
    {
    	$P_{i,0} = shortestPath(s_i,g_i)$\;
    }
    \While{some agent is not at its goal}
    {
        partition robots into inverse chains\;
		\For{each inverse chain} {
		\if{the inverse chain contains a (nontrivial) cycle}
		{move all the robots in the cycle\;
		}
		\else{
        advance all the robots in one maximal chain\;}
        }
 	\For{$i=1,\dots,k$}
  	{	
    	\If{$r_i$ is unmoved and there is an unmoved robot $r_j$ occupying $r_i$'s next node}
        {
          \If{$P_{j,t}\subset P_{i,t}$}
          {
          swap($r_i,r_j$)\;
          }
    	}
   	}
}

\end{algorithm}

\subsection{Bubbletree Optimizations}
\subsubsection{Forming the tree}
There are various characteristics that we might want the bubbletree to have
\begin{enumerate}
\item Distance ``preserving'', especially between pairs of starts and goals
\item Low diameter
\item Low number of nodes
%\item Maximize the number of pairs $\{g_i,s_i\}$ that appear in the same depth-1 subtree
\end{enumerate}
\par %I don't know a good way to preserve distances -- http://www.worldscientific.com/doi/abs/10.1142/S1793830913500109 is refusing to load T.T
 Unfortunately, it is impossible to restrict to a tree while preserving distances to within a constant factor, as is demonstrated by polygon graphs.
\par (Hassin and Tamir 1995) presents a polynomial method to form a spanning tree of (nearly) minimum diameter. Simply form a breadth-first-search tree from the ``one-center'', the node whose max distance $d$ to all other nodes is minimized. It is clear that the minimum diameter must be at least $d$. And the breadth-first-search tree has diameter at most $2d$.
\par For minimizing the number of nodes, it is useful to use Steiner trees. For a weighted-graph $G$ with a subset $S\subset V$ of the vertices, the generalized Steiner tree problem is to find a subtree $T$ containing all the vertices in $S$, such that the sum of the edge weights is minimized. It is possible to approximate the minimum to within a factor of $\text{ln}(4)$ (Byrka et al. 2010), although approximating to within a factor of $96/95$ is NP-hard (Chleb{\'i}k et al. 2008) %both cites are actually from https://en.wikipedia.org/wiki/Steiner_tree_problem lol

%It's impossible to preserve distances, however, except for, um, trees and, I don't know
%I think low diameter is doable, as is low number of nodes
%The makespan is $O(d+k)$, where $d$ id the diameter of the bubbletree, so it makes sense to minimize the diameter. 
%Independence detection is a very successful technique in MAPF that can be performed on top of any solver. It works by combining robots into a meta-
%ID doesn't seem super helpful here because those paths are considered in the time-expanded graph. What I need is prior results on disjoint spanning (or covering?) forests that are guaranteed to contain certain subsets
\subsubsection{Bubble Forest}
\par We can get an even bigger improvement, especially on maps with relatively few robots, by forming a bubble forest. (Lopez 2006) presents a Steiner forest, which can be found in polynomial time by solving a dual linear program and is able to %Adriana only cites Vazirani, Vijay V. "Approximation Algorithms" Chapters 3.1 and 22, Springer 2003
\begin{itemize}
\item Ensure each $\{s_i,g_i\}$ pair is in the same tree
\item minimize the total number of edges to within a factor of $2$ (in fact, the algorithm works more generally for edges with positive-rational length.) %or maybe the number of nodes...
%\item
\end{itemize}
Of course, minimizing the number of edges also does a lot to minimize the number of nodes, as the typical PERR formulation is on a graph with unit-cost edges.

% \item Choosing $v_m$ so that its subtrees are goal-balanced, rather than node-balanced.
% \item Moving as many $M_c$ as possible in a given timestep, rather than just the focal $M_c$.
% \item Devoting some time at the start of each recursive call of bubbletree to bring all robots to within $k'$ of $v_m'$ (here $k'$ and $v_m'$ denote the number of robots and mid node in the recursive call.) This way we effectiveley have $l'\leq k$.

\subsubsection{Goal-balancing $v_m$}
To get a linear makespan, we need to show that the problem shrinks in size by a constant factor (half) with each recursive call. In the basic bubbletree algorithm, the main measure of the problem size is the number of nodes. To prove the $O(k+l)$ result, however, we need to express the problem size in terms of $k$. Thus we choose the mid node $v_m$ such that each of the subtrees defined by its children has at most $k/2$ goal vertices. The procedure is very similar to that in basic bubbletree: perform a breadth-first search from an arbitrary vertex $v$, and mark every vertex $v'$ with the number of goal vertices in the subtree it defines, $g(v')$. If $g(c)\leq k/2$ for every child $c$ of $v$, then we are done. Otherwise let $d$ be the unique goal node for which $g(d)>k/2$. Then $d$ is a satisfactory mid node.

\subsubsection{Moving all Migrant Sets}
In the base algorithm, we perform two main (non-recursive) operations: advancing the focal $M_c$ in priority order, and advancing a chain in $G^T_c$ with root $v_m$. We can simultaneously move other migrants, subject to ceding priority to the focal $M_c$ and then chain, and maintaining the priority orderings within each of the migrant sets. %ugh, should I be deleting and adding to the G^I's when robots leave and enter?
\par Refer to bubbletree with the joint optimizations of goal-balancing together with moving all migrant sets as bubbletree2.
\par \begin{proposition}  Bubbletree2 produces a makespan of $O(\text{diam}*\text{log}k+k)$. %down to $O(diam+k+deg)$, but we can easily prune so deg\leq k, great. Ugh, so that deg(v_m)\leq k. Anyway, every path needs to terminate in a start or goal, so deg\leq 2k, fine.
\end{proposition}
\begin{proof} The proof uses the following lemma. 
\begin{lemma} After $t$ timesteps, the $i^{th}$ highest priority robot in any $G^I_c$ (the highest priority is $i=0$) is already an infiltrator or else is at most $\text{max}\{\text{diam}+2i-t+2k,i+1\}$ steps away from $v_m$.
\end{lemma}
\begin{proof}
$2i-t$ is as in the proof of the guarantee for the basic bubbletree algorithm, $2k$ comes from the fact that we spend at most $2k$ timesteps moving $r_m$, and $i+1$ comes from the new possibility that a line of migrants is waiting for $v_m$ to empty. We can thus ignore the $2k$ term. The induction proceeds with the exact same cases as in the proof of the guarantee for basic bubbletree. Clearly the base case holds. We will reproduce the rest of the argument here.
\par Assume we have the stated invariant at timestep $t$ (and ignore the $2k$ term). It is still the case that infilitrators are never pulled back across $v_m$. Consider what it means if some migrant $r_i\in M_c$ that has not yet gotten to $v_m$ fails to advance. There are the same two old possibilities, and also a new possibility. The first is if a higher priority robot $r_j$ swaps another robot to the node that $r_i$ was planning on advancing to. Then $r_i$ is only one step worse off than $r_j$ was the turn before at $t-1$. The guaranteed proximity to $v_m$ for $r_j$ was $\text{max}\{\text{depth}(T_c)+2(i-1)-(t-1),(i-1)+1\}=\text{max}\{\text{depth}(T_c)+2i-t-1,i\}$. So the guarantee for $r_i$ is  $\text{max}\{\text{depth}(T_c)+2i-t,i+1\}$ as desired. The second is that a higher priority robot $r_j$ bully-swapped it. This bodes even better for $r_i$ than if it is snubbed, and the guaranteed closeness to $v_m$ is again maintained. The new possibility is if $r_i$ is stymied by the robot occupying $v_m$. Then $r_i$ is at most one step from $v_m$, and we are good.
\end{proof}
By the lemma, every migrant is within $k$ of $v_m$ after $\text{diam}+4k$ timesteps. We can then invoke the basic bubbletree guarantee, since the optimization honors the priorities from basic bubbletree. The effective depth of each subtree $T_c$ is at most $|M_c|$, so we can finish off the pre-recursive phase in at most $\sum_{c}|M_c|+\sum_{c}2|M_c|+2k$ more steps (recall the guarantee from the basic analysis is $\text{depth}(T_c)+2|M_c|$ per tree $T_c$, plus $2k$ steps for $r_m$.) So the pre-recursive phase takes at most $\text{diam}+10k$ timesteps.
\par When we recurse, the number of robots decreases by at least half. Unfortunately, we have no such guarantee on the diameter (or the longest path length). Thus we solve the recursion $f(\text{diam},k) \leq O(k+\text{diam})+f(\text{diam},k/2)$ and get $f(\text{diam},k) = O(\text{diam}*\text{log}k+k)$ as claimed.

\end{proof}

\subsubsection{Advancing Infiltrators}
\par We can also advance robots that are already in their target trees. We refer to these non-migrants as infiltrators (they include all non-migrants other than $r_m$.) Since the infiltrators have the lowest priority, it is not very important how to move them. A simple solution is to initialize an ordering on each $G^T_c$ and move the infiltrators in priority order. A more complicated approach is to try move the infiltrators in $T_c$ out of the way of the migrants in $T_c$.
\par  What we see as the most natural approach is to preemptively do the actions that would be taken after the pre-recursive phase. Let $T^i_c$ denotes a subtree in an $i^{th}$ recursive call of bubbletree. Thus there is one $T^0_c$, and it is the entire bubble tree. The $T^1_c$'s are the familiar subtrees defined by the children of $v_m$. We say a $T^i_c$ is order $i$. We can similarly define the higher order migrant set $M^i_c$ as the robots in $T^i_c$ whose goals are not in $T^i_c$. This way we have $M^i_c\subset M^{i+1}_{c'}$ when $c'$ is a descendant of $c$. In each timestep, loop through the $T^i_c$'s from lowest to highest order, and move the robots in $T^i_c$ as in bubbletree2, treating robots in strictly lower-order migrant sets as immovable obstacles. There are only $O(k)$ $T^i_c$'s, so this optimization is not excessively time-consuming. %higher order migrant sets. Let $T^i_c$ denote a subtree in an $i^{th}$ recursive call of bubbletree. In addition to finding $v_m$, we can find all the mid nodes $v^i_m$ of the various $T^i_c$ subtrees when we first form the initial bubble tree. We can similarly define the higher order migrant set $M^i_c$ as the robots in $T^i_c$ whose goals are not in $T^i_c$. This way we have $M^i_c\subset M^{i+1}_{c'}$ when $c'$ is a descendant of $c$. Say a robot is order $i$ if it is in $M^{i+1}_{c'}\setminus M^i_c$, with $c'$ a descendant of $c$.  %It makes sense to give priority to the lower order robots, and to determine how to move them based on the robots occupying lower-order $v^i_m$'s. In each timestep, the optimization thus works by moving all robots as in bubbletree2, then moving 
%This scheme removes the need for recursive calls altogether, as we know that the highest order of any $v^i_m$ is $\text{lg}k$ (assuming goal-balancing the mid nodes).%that are in lower-order migrant sets. %Unfortunately, this can result in local cycles. For example, if tw%After looping through the zeroth order migrant sets, we can then loop through the first order migrant sets%If we node-balance rather than goal-balance, then we can guarantee that every $v^{i+1}_m$ in a subtree $T^i_c$ (with mid node $v^i_m) is farther from $v_m$ than $v^i_m$ is. This way we have $M^i_c\subset M^{i+1}_{c'}$ when $c'$ is a descendant of $c$. Unfortunately, my conception of higher order migrant sets can lead to local cycles %There are a couple conceivable ways to do this. One is to begin the operations that will commence once bubbleTree is recursively called on that subtree. %However, this interferes with runtime parallelization (sort of), has robots clogging mid (I think), requires preemptively deciding on a leaf-hierarchy (which is presumptuous), 
\par Using the additional optimization of advancing infiltrators to bubbletree2 gives an $O(k+l)$ solution to the counterexample that proved the necessity of the $\text{log}k$ factor in bubbletree2. However, we do not know if it gives an $O(k+l)$ solution in all cases. In particular, the longest path to goal (or the diameter) need not decrease by any constant factor in the first recursive call. For example, they would not decrease if there was an order $1$ robot $r_1$ whose path to goal was $l$ (with $l$ nearly the size of the diameter): we could start with $r_1$ at a grandchild of $v_m$, and engineer it so that there is always an order $0$ robot occupying that grandchild's parent, or else vying for it.

%\par 
%\subsubsection{Making $d_{depth}'\leq k$}
%Devoting some time at the start of each recursive call of bubbletree to bring all robots to within $k'$ of $v_m'$ (here $k'$ and $v_m'$ denote the number of robots and mid node in the recursive call.) This way we effectiveley have $l'\leq k$.

%\begin{theorem} The three optimizations -- goal-balancing, moving all $M_c$'s every timestep if possible, and making $d_{depth}'\leq k$ -- give an $O(l+k)$ makespan.
%\end{theorem}
%\begin{proof}

%\end{proof}
\subsubsection{Recomputing Bubble Trees}
When we recursively call bubble tree on a subtree, we can instead pass the nodes of the tree and compute a new bubble tree (or forest) on those nodes. %This optimization clashes with some attempts to advance infiltrators, however.

\subsubsection{Earlier Recursive Calls}
We can also call bubbletree on each $T_c$ as soon as the robots occupying $T_c$ are precisely the robots in $G^T_c$.

% \subsubsection{postprocessing}%this might be redundant with dynamic parallelization
% We can keep flags of the last time that any robot enters or leaves each $T_c$. In postprocessing, we can retroactively call bubbleTree on each $T_c$ as soon as we realize the robots in it are precisely the group $G^T_c$. %For example, this will speed things %This can happen, for example, -- actually, I don't think this can happen lol
% %Also, we can do this on a per robot basis, but we implicitly do it anyway by looping, so nevermind

\subsubsection{Leafhunter Hierarchy} We can be intelligent about how we assign priority within each $G^I_c$ (and each $G^T_c$), rather than doing it randomly. We can require higher priority for robots that will have farther to go after becoming infiltrators. More minimally, we can insist on higher priority for $r_i$ than $r_j$ if $g_i$ is a descendant of $g_j$. %, if we have $\{r_i,r_j\}\subset G^I_c\cap G^T_d$ and $r_i$ is a descendant of $r_j$ in $T^d$, then $r_i$ has higher priority.
%\par In addiction to moving the migrants,we can have the infiltrators bury themselves in their target trees, ideally to their goal nodes. I call a robot whose goal is a leaf a leaf-hunter. Such a robot does not interfere with other robots once it has met its goal. Removing leaves, for example those occupied by leaf-hunters, opens up new leaf-hunters. In this way we can form a leaf-hunter hierarchy. %Indeed, it's not leaves per se so much as not interfering with other guys in the same target class. Ok, there's still the problem of what to do if a robot has to turn around and go up its (sub)subtree. 
\par Using the leaf-hunter hierarchy also gives a means for smarter chain resolution. In the base algorithm, we choose the chain to advance randomly. If we prioritize leaf hunters, then at least the structure of these chains will not be completely random. Rather, the deeper members of these chains are more likely to be leaf-hunters, and therefore be relatively better off heading the push into the depths of the subtree.

\subsubsection{Postprocessing} With some of these more complicated optimizations (in particular, advancing infiltrators), there is a risk that we introduce redundant movement. For example, we may advance an infiltrator up to $v_m$, only to have it swapped back the next turn by a higher-priority migrant. There is a redundancy precisely when $r_i$ occupies $v$ at times $t$, vacates $v$ at time $t+1$, and returns at time $t+c$ without any other robot occupying $v$ in the meantime. A general path-planning optimization is to go through in postprocessing and replace this redundant traveling with waits (cite the push and swap paper.) It may be necessary to loop through the solution multiple times, as new moves may become redundant after deleting old redundant moves.

%\subsubsection{Pseudocode for Fully Optimized Bubbletree}
%you should probably handwrite this first

%\subsubsection{Smarter chains}

% \subsubsection{using multiple computers for runtime gains} Of course, we can also compute the paths in parallel in $O((d+k+n)k) = O(nk)$ time, using that each timestep takes $O(k)$ time (plus accessing the graph, which I treat as constant time). %O(n), ok. But damn, just reading the graph is actually the hardest part. No, finding all shortest paths can be O(kn) in the worst case.
% Performing all the computations in the main loop on one computer gives the recursion $f(n) = c(d+k)k+2f(n/2)$, which has solution $O(nk\text{log}nk)$ for dense bubbleTrees. %dude, you're assuming each timestep is unit cost, whereas it's really cost k (plus accessing the graph, ok)
% Recall that finding all shortest paths can be $O(kn)$ in the worst case, thus we improve the asymptotic runtime by computing the paths in parallel.

% \subsubsection{bounded degree trees} Can use the phase-taking target-centric algorithm that was more directly inspired by bubblesort.

\section{Comparison Between RIP and BubbleTree}
%RIP is provably good when $k^2<<\text{SIC}$, where $\text{SIC}=\sum_{i=1}^kd(s_i,g_i)$ is the sum of individual costs. 
RIP is easier to implement. %(i.e. I have already implemented it lol.) %Right, I need to see the experimental results before I can say how RIP outperforms theory. Conceivably, the bubble tree could also distort the distances. However, if it is implemented as a Steiner forest, then 
RIP can also outperform bubbletree on cyclic graphs e.g. on the cycle graph at the end of the basic bubbletree makespan analysis, a graph on which RIP performs optimally. RIP also performed near-optimally in the DAO map, the 20x15 grid maps, and the dense linear arrays. 
\par Extrapolating RIP's performance on the dense square grids suggests that bubbletree would outperform RIP on robot-dense two-dimensional arrays. Bubbletree also has much stronger theoretical guarantees.

\section{Discussion}
We went through a lot of work to show that RIP produces a makespan that is $O(\text{SIC}+k^2)$. Empirically, RIP performs much, much better. And we have not found any examples that are $O(\text{SIC}+k^2)$ but not $O(k^2)$. In fact, it could still be that optimizations of RIP obtain an $O(n)$ or even an $O(k+l)$ makespan! (Where $l$ is the maximum distance from start to goal.)
\par Why is it so difficult to prove a stronger result? It is very natural to modify the potential function. For example, we could include terms for the number of times pairs of paths intersect, or we could somehow amortize how often a robot is stymied by another robot, including stymied by a robot at the end of its chain. A natural conjecture is that $r_i$ only swaps $r_j$ once. But we can in fact engineer scenarios in which $r_i$ swaps $r_j$ arbitrarily many times. It is sufficient to show that we can pull $r_i$ and $r_j$ off of their goals and force them to subset-swap. To manage this, start with $r_i$ at $(0,0)$ and $r_j$ at $(1,0)$ and with two robots above $r_j$. Have those two robots subset-swap $r_j$ up to $(1,2)$, and then have a stream of robots file through $(1,1)$, preventing $r_j$ from moving. We can then have three robots pull $r_i$ over to $(3,0)$. We can then free $r_j$ to its goal. But now there is a swap need $(i,j)$. How to amortize these machinations is beyond me.
\par Maybe $r_i$ can keep a labeled tally of the times it has been stymied, and share that tally with any robot it subset-swaps. Intuitively it makes sense that, if $r_i$ swaps $r_j$ at time $t$, then $r_j$ is at least as well off at time $t+1$ as $r_i$ was guaranteed to be at time $t$. But what exactly to tally, and how to tally it, is unclear. 
\par Another approach is to use inverse chains. We can modify RIP so that the largest subchain of an inverse chain is guaranteed to advance every other timestep. But what if the trees are shallow and branch rapidly? Another approach is to select the subchain containing the robot that still has the farthest to go -- this way the maximum distance to goal can be tracked and shown to consistently decrease. But, there can be many robots that are the same distance to goal, and we can only be guaranteed to advance one of them -- and only every other turn, at that. %An even greater dilemma is what to do if
\par Really what we need is a linear measure of global progress. Of course, this is the motivation for trying a different algorithm, like bubbletree.

\section{Acknowledgments}
Thank you to USC's SURE program for funding this research. Thank you to Guni Sharon for making the code of his CBS solver available to us. Thank you to Sven Koenig and Hang Ma for working with me in their lab, including: presenting me with PERR and the pseudocode for RIP, as well as the idea of using a monovariant; giving me an overview of their prior research; helping with the code and Excel; talking through my proofs and ideas with me; and securing funding through SURE.

% \section{References}
%Hang, Steiner tree stuff and one-center, push and swap postprocessing.
%Also, background like cbs, Szemeredi's expander graphs, odd-even sort, 

\end{document}